\documentclass[11pt]{amsart}
\usepackage[utf8]{inputenc}
\usepackage{amsmath}
\usepackage{amsthm}
\usepackage{graphicx}
\usepackage{amsfonts}
\usepackage{amssymb}
\usepackage{bbm}
\usepackage{fullpage}
\usepackage{verbatim}
\usepackage{color}
\usepackage{datetime}
\usepackage{graphicx}
\usepackage{cite}
\usepackage{bbm}
\usepackage{hyperref}
\usepackage{tikz}
\usetikzlibrary{calc}
\usepackage{caption}
\usepackage{subcaption}

\newcommand\hexsize{16pt} 

\numberwithin{equation}{section}

\newcommand{\rr}{\m{R}}

\newcommand{\cc}{\m{C}}

\newcommand{\1}{\mathbbm{1}}

\newcommand{\la}{\lambda}

\newcommand{\ep}{\varepsilon}

\newcommand{\mr}{\mathring}

\newcommand{\set}[1]{\left\{ #1 \right\} }
\newcommand{\ip}[1]{\langle #1 \rangle}

\newcommand{\norm}[1]{\left \Vert #1 \right\Vert }

\newcommand{\m}[1]{\mathbb{#1}}
\newcommand{\mc}[1]{\mathcal{#1}}
\newcommand{\mf}[1]{\mathfrak{#1}}

\newtheorem{theorem}{Theorem}[section]
\newtheorem{lemma}[theorem]{Lemma}
\newtheorem{proposition}[theorem]{Proposition}
\newtheorem{corollary}[theorem]{Corollary}
\newtheorem{definition}[theorem]{Definition}

\newtheorem{assumption}[theorem]{Assumption}

\newcommand{\dom}{\mathrm{dom}}

\newcommand{\bR}{{\mathbb R}}
\newcommand{\bC}{{\mathbb C}}
\newcommand{\bN}{{\mathbb N}}
\newcommand{\bZ}{{\mathbb Z}}
\newcommand{\cA}{{\mathcal A}}
\newcommand{\cB}{{\mathcal B}}

\newcommand{\cG}{{\mathcal G}}
\newcommand{\cH}{{\mathcal H}}
\newcommand{\cP}{{\mathcal P}}
\newcommand{\cS}{{\mathcal S}}

\newcommand{\caA}{{\mathcal A}}
\newcommand{\caB}{{\mathcal B}}
\newcommand{\caC}{{\mathcal C}}

\newcommand{\caG}{{\mathcal G}}
\newcommand{\caH}{{\mathcal H}}
\newcommand{\caI}{{\mathcal I}}

\newcommand{\caP}{{\mathcal P}}

\newcommand{\caS}{{\mathcal S}}
\newcommand{\caT}{{\mathcal T}}

\newcommand{\caV}{{\mathcal V}}
\newcommand{\caW}{{\mathcal W}}

\newcommand{\braket}[2]{\langle {#1}, {#2}\rangle}

\newcommand{\A}{\mathcal{A}}

\newcommand{\supp}[1]{\textnormal{supp}(#1)}
\newcommand{\HH}[1]{\mc{H}^{(#1)}}
\newcommand{\su}{\mf{su}(2)}
\newcommand{\BOmega}{\mathbf{\Omega}}

\newcommand{\Ent}[2]{D_{\infty}\left( #1 \middle\| #2 \right)}

\usepackage{centernot}
\newcommand{\disjoint}{\centernot\mid}

\title{Stability of the spectral gap and ground state indistinguishability for a decorated AKLT model}
\author{Angelo Lucia, Alvin Moon and Amanda Young}

\begin{document}
\maketitle

\begin{abstract}
  We use cluster expansion methods to establish local the indistiguishability of the finite-volume ground states for the AKLT model on decorated hexagonal lattices with decoration parameter at least 5. Our estimates imply that the model satisfies local topological quantum order (LTQO), and so the spectral gap above the ground state is stable against local perturbations. %
\end{abstract}

\section{Introduction}

Quantum phases of matter are equivalence classes of systems which share similar physical properties.
A central question in the study of quantum many-body systems is to determine the phase to which a given models belongs, and one of the fundamental quantities in this classification is the spectral gap above the ground state energy. For example, Haldane predicted there would be distinct gapped and critical phases for antiferromagnetic spin chains \cite{Haldane1983a,Haldane1983}. Furthermore, the existence of a bulk gap in the presence of gapless edge excitations is the defining characteristic of topological insulators. Under rather general conditions, a non-vanishing gap also implies exponential clustering of the ground state \cite{HK, NS2}. Much of the recent focus has been on studying topological phases of matter, including symmetry-protected phases \cite{GW, PhysRevB.83.035107, Chen2013, PTBO, PTBO2, Sopenko2021, Maekawa2022, Tasaki2021}. A key element for defining topological indices \cite{ogata2019mathbb, Ogata2019, ogata2021h3gmathbb} is the split property, which for one-dimensional systems is known to hold if the system is short-ranged and gapped \cite{M}. Thus, the classification of gapped ground state phases is of particular interest.

Two models are said to belong to the same gapped phase if the interactions can be smoothly deformed into one another without closing the spectral gap. Despite its importance, it is generically undecidable (in an algorithmic sense) to determine rigorously whether or not a one- or two-dimensional translation-invariant, frustration-free, nearest-neighbor quantum spin model has a non-vanishing gap\cite{Cubitt2015a, Cubitt2015, Bausch2018}. Thus, a natural approach is to analyze the properties of a known gapped model. In order for the spectral gap to be physically relevant, it needs to be robust against noise; that is, small perturbations of the model should still belong to the same phase. If this is not the case, then it is unlikely that a phase representing this model will be observed experimentally, and so the stability of the spectral gap is also fundamental to the study of quantum phases.

The \emph{quasi-adiabatic continuation} (also called the spectral flow) introduced by Hastings and Wen in \cite{Hastings2005} has proved to be an invaluable tool for exploring gapped ground state phases \cite{BMNS, BO,BDRF1, Ogata2019,ogata2021h3gmathbb,Moon2020, Moo, MN, nachtergaele2021stability}. In \cite{BH, BHM}, Bravyi, Hastings and Michalakis (BHM) pioneered a general strategy that utilizes this automorphism to prove spectral gap stability for quantum spin models with commuting interactions. This was extended to frustration-free interactions by Michalakis and Zwolak in \cite{MZ}, and further developed in a number of directions, including to systems with discrete symmetric breaking and topological insulators, by Nachtergaele, Sims and Young in \cite{NSY, NSY:2022, NSY:2021}. The BHM strategy shows that the spectral gap is stable against sufficiently local perturbations if the ground states satisfy a property known as \textit{local topological quantum order} (LTQO). Roughly speaking, LTQO holds if the finite-volume ground states cannot be distinguished by any local operator acting in the bulk. For this reason, LTQO is also referred to as \emph{local indistinguishability} of the ground states. The stability of the decorated AKLT model in this work will be proved using the BHM strategy, and in particular, by applying the result from \cite{NSY:2021}.

We comment that an alternative approach to proving spectral gap stability based on Lie-Schwinger diagonalization was developed in \cite{FP,DelVecchio2021,Delvecchio:2022}. This technique applies to unperturbed models with product ground states of non-interacting systems, which trivially satisfy LTQO. A strength of this approach is that it also apply to models with unbounded terms as well as non-self adjoint Hamiltonians. Methods for gap stability of quasi-free lattice fermion models have also been investigated \cite{Koma2020, Hastings2019, DeRoeck2019}.

Beyond spectral gap stability, local ground state indistinguishability has been used to study the stability of other physical properties, including super-selection sectors \cite{Cha2020}, vanishing Hall conductance \cite{Zhang2022}, and the stability of invertible states \cite{bachmann2021stability}. %
Related to this, it was shown in \cite{henheik2021local, bachmann2021stability} that local perturbations of a certain class of weakly-interacting, gapped systems at most perturb ground states locally, even if the perturbation closes the gap. A variation of LTQO for open systems was also used to prove stability of dissipative systems with unique fixed points \cite{Cubitt2015b}.  As suggested by these results, proving ground state indistinguishability may be of independent interest.

The ground states of the $SU(2)$-invariant antiferromagnetic models introduced by Affleck-Kennedy-Lieb-Tasaki (AKLT) \cite{AKLT87,AKLT88} have served as an important case study for many questions in quantum many-body physics. Recently, significant progress was made for the long-standing spectral gap conjecture of the AKLT model on the hexagonal lattice. Decorated versions of this model defined by replacing each edge of the hexagonal lattice with a spin chain of length $d$ (see Figure~\ref{fig:DecoratedLattice}) were considered in \cite{BRLLNY}. It was proved there that these models have a uniform gap for decoration parameters $d\geq 3$ and finite volumes suitable for periodic boundary conditions. Subsequent results based on combining the analytical methods from \cite{BRLLNY} with a Lanczos numerical method established additional gap results for two-dimensional AKLT models \cite{Pomata_2020,Pomata2019}, including the (undecorated) hexagonal model. An independent result based on using DMRG to verify a finite-size criterion simultaneously appeared  \cite{Lemm_2020}. These constitute some of the few examples of 2D models with non-commuting interactions for which rigorous gap estimates have been obtained.

In this work, we take the next step and consider the stability of the spectral gap for the decorated AKLT models on the hexagonal lattice. We show that these models have indistinguishable ground states that satisfy the LTQO condition for decoration parameters $d\geq 5$, and therefore belong to a stable gapped phase. This answers positively one of the open questions raised in \cite{BRLLNY}. To the best of our knowledge, this is the first rigorous proof of a stable gapped phase for a non-commuting two-dimensional interaction. 

We prove stability of the spectral gap in the infinite volume setting by verifying the conditions of \cite[Theorem 2.8]{NSY:2021} under mild modifications to account for the specific geometry of the hexagonal lattice. The LTQO condition is an immediate consequence of our indistinguishability result, which is proved using a uniformly convergent cluster expansion of the ground state expectations that is given in terms of a hard core gas of loops and walks. We closely follow the cluster expansion strategy  used to study the undecorated hexagonal AKLT model in \cite{Kennedy1988}. The novelty here is new estimates on the rate of convergence of the finite volume AKLT ground states to the frustration-free bulk state. These estimates make explicit how the convergence depends on size of the support and operator norm of the local observable considered, which is vital for establishing LTQO.

Given that the model has a spectral gap for all values of the decoration, including the case $d=0$, it remains an open problem to show the LTQO condition when $d<5$. While (significantly) tighter counting arguments and estimates on the cluster expansion could in principle prove the result for lower values of $d$, we suspect that the strategy used here would not extend all the way to $d=0$, and so we anticipate (just as in the case of the spectral gap proof) that a varied approach would be needed in that regime.  A different method using cluster expansions for a more restrictive class of perturbations was used to prove stability for the one-dimensional AKLT model in \cite{Y}. This result takes advantage of the fixed ground state degeneracy of the one-dimensional model, though, which does not hold for the multidimensional models, and so it is unclear if this approach can be adapted to the present setting.

This paper is organized as follows. In Section~\ref{sec:main_results}, we introduce the decorated lattice models and summarize our main result. In Section~\ref{sec:stability}, we present the modified stability conditions of~\cite{NSY:2021} needed to prove stability of the gap, and show that they hold under the assumption that the ground states satisfy the indistinguishability result Theorem~\ref{thm:indistinguishable}. In Section~\ref{sec:gss}, following the method of \cite{Kennedy1988} we represent the ground states on the decorated lattices in terms of a hard-core gas of loops and walks. Finally, in Section~\ref{sec:main} we use a cluster expansion argument to estimate the convergence of an arbitrary finite volume ground state to the unique infinite volume frustration-free state, proving the indistinguishability result. Technical lemmas and counting arguments for the loop models are contained in the Appendix~\ref{sec:appendix}.

\section{Setup and main results}\label{sec:main_results}

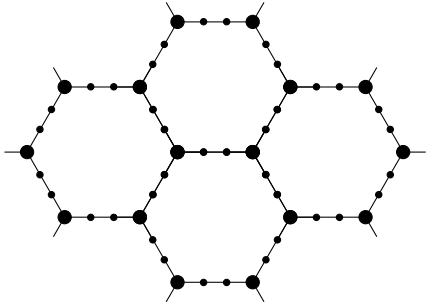
\begin{figure}
	\begin{tikzpicture}
		\newcommand{\hexcoord}[2]
		{[shift=(0:#1),shift=(60:#1),shift=(0:#2),shift=(-60:#2)]}
		\foreach \x in {0,...,1}
		\foreach \y in {0,1}{
			\draw\hexcoord{\x}{\y}
			(0:1)--(60:1)--(120:1)--(180:1)--(-120:1)--(-60:1)--cycle;
			\foreach \n in {0, 60, 120, 180, -120, -60}{
				\draw\hexcoord{\x}{\y} (\n:1)--(\n:1.3);
				\draw[fill=black]\hexcoord{\x}{\y} (\n:1) circle (2.5pt);
				\draw[fill=black]\hexcoord{\x}{\y} (\n+20:{sqrt((2*cos(\n)+cos(\n+60))^2/9+(2*sin(\n)+sin(\n+60))^2/9)}) circle (1.2pt);
				\draw[fill=black]\hexcoord{\x}{\y} (\n+40:{sqrt((cos(\n)+2*cos(\n+60))^2/9+(sin(\n)+2*sin(\n+60))^2/9)}) circle (1.2pt);
			}
		}
	\end{tikzpicture}
	\caption{The decorated hexagonal lattice for $d=2$.}	
	\label{fig:DecoratedLattice}
\end{figure}
The AKLT model consider in this work is defined on the $d\in\bN_0$ decorated hexagonal lattice $\Gamma^{(d)} := (\caV^{(d)},\caB^{(d)})$, which is the graph resulting from appending $d$ additional sites to each edge of the hexagonal lattice, see Figure~\ref{fig:DecoratedLattice}. Here, $\caV^{(d)}$ denotes the set of vertices (or sites) of the decorated lattice, and $\caB^{(d)}$ is the set of bonds (or edges). The hexagonal lattice, $\Gamma^{(0)}$, will also be called the undecorated lattice. Throughout this work, we will consider finite subsets of the decorated lattice, and it will be convenient to consider the subvolumes as graphs. As such, we begin by introducing the graph notation that will be used throughout the work, and then review the notation for the model of interest and state the main results.

For any subgraph, $\Lambda = (\caV_\Lambda,\caB_\Lambda)$ of $\Gamma^{(d)}$ with vertex set $\caV_\Lambda$ and edge set $\caB_\Lambda$, we denote by $|\Lambda|$ the number of vertices $|\caV_\Lambda|$, and when this quantity is finite we will call $\Lambda$ a finite volume. With a minor abuse of notation we simply write
$v\in \Lambda$ and $(v,w)\in \Lambda$ to denote a vertex $v$ and edge $(v,w)$ of $\Lambda$, respectively. The degree of a vertex $v$ in $\Lambda$ will be denoted by  $\deg_{\Lambda}(v)$, and we simply write $\deg(v)$ for $\deg_{\Gamma^{(d)}}(v)$. Note that $\deg_{\Lambda}$ necessarily takes values in $\{0,1,2,3\}$ while $\deg$ takes values in $\{2,3\}$.
We define the \textit{graph union} $\Lambda_{1}\cup \Lambda_{2}$ of two subgraphs $\Lambda_{1}$, $\Lambda_{2}$ as the graph with vertices $\caV_{\Lambda_{1}} \cup \caV_{\Lambda_{2}}$ and edges $\caB_{\Lambda_{1}}\cup\caB_{\Lambda_{2}}$.
The set of boundary sites $\partial \Lambda$ is the collections of all sites $v\in \Lambda$ belonging to an edge $(v,w)\in\caB^{(d)}\setminus \caB_\Lambda$ that leaves $\Lambda$, and $\mr{\Lambda}$
denotes the subgraph of $\Lambda$ induced by the interior sites $\caV_{\Lambda}\setminus\partial \Lambda$, namely the graph with vertex set $\caV_{\Lambda}\setminus \partial \Lambda$ and containing only the edges of $\Lambda$ which are not incident to vertices in $\partial \Lambda$.

\begin{figure}
	\begin{center}
		\begin{tikzpicture}
			\def\n{.55};
			\newcommand{\hexcoord}[2]
			{[shift=(60:#1),shift=(120:#1),shift=(0:#2),shift=(-60:#2),shift=(0:#2),shift=(60:#2)]}
			\foreach \x in {3,...,7}
			\foreach \y in {1,...,6}{
				\draw[color=black!25]\hexcoord{\x*\n}{\y*\n}
				(0:\n)--(60:\n)--(120:\n)--(180:\n)--(-120:\n)--(-60:\n)--cycle;
				\draw[color=black!25,shift=(-60:\n), shift=(0:\n)]\hexcoord{\x*\n}{\y*\n}
				(0:\n)--(60:\n)--(120:\n)--(180:\n)--(-120:\n)--(-60:\n)--cycle;
			}
			\foreach \x in {3,...,8}{
				\draw[color=black!25,shift=(-60:\n), shift=(0:\n)]\hexcoord{\x*\n}{0}
				(0:\n)--(60:\n)--(120:\n)--(180:\n)--(-120:\n)--(-60:\n)--cycle;}
			\foreach \y in {1,...,6}{\draw[color=black!25,shift=(-60:\n), shift=(0:\n)]\hexcoord{8*\n}{\y*\n}
				(0:\n)--(60:\n)--(120:\n)--(180:\n)--(-120:\n)--(-60:\n)--cycle;
			}
			
			\foreach \k in {0,3}
			\foreach \x in {0,60,120,180,240,300}
			\foreach \y in {0,60,120,180,240,300}
			{\draw[shift=(\x+120:\n), shift=(\x+60:\n), very thick, color=black]\hexcoord{5*\n}{5*\n-\k*\n}
				(\x+\y:\n)--(\x+\y+60:\n);
				\draw[shift=(\x+120:\n), shift=(\x+60:\n), fill=black]\hexcoord{5*\n}{5*\n-\k*\n} ({\n*(cos(\x+\y)+cos(\x+\y+60))/2},{\n*(sin(\x+\y)+sin(\x+\y+60))/2}) circle (1.1pt);
				\draw[shift=(\x+120:\n), shift=(\x+60:\n), fill=black]\hexcoord{5*\n}{5*\n-\k*\n} ({\n*(cos(\x+\y)},{\n*(sin(\x+\y)}) circle (1.6pt);
			}
			
			\foreach \x in {0,60,120,180,240,300}
			\foreach \y in {60,120}{
				\draw[shift=(\x+120:\n), shift=(\x+60:\n), color=red, fill=red]\hexcoord{5*\n}{2*\n} ({\n*(cos(\x+\y)},{\n*(sin(\x+\y)}) circle (1.6pt);}
			
			\foreach \x in {0,60,120,180,240,300}
			\foreach \y in {60,120}{
				\draw[shift=(\x+120:\n), shift=(\x+60:\n), very thick]\hexcoord{5*\n}{5*\n} (\x+\y:\n)--(\x+\y:1.5*\n);
				\draw[shift=(\x+120:\n), shift=(\x+60:\n), fill=black]\hexcoord{5*\n}{5*\n} (\x+\y:1.5*\n) circle (1.1pt);}
			
			\draw\hexcoord{5*\n}{2*\n} (60:\n)--(120:\n) node[midway, below, yshift=-2pt] {$\Lambda_2^{(1)}$};
			\draw\hexcoord{5*\n}{5*\n} (60:\n)--(120:\n) node[midway, below, yshift=-2pt] {$\Gamma_2^{(1)}$};
		\end{tikzpicture}
	\end{center}
	\caption{ Illustration of $\Lambda_2^{(1)}$ and $\Gamma_2^{(1)}$. The latter is used to verify Assumption~\ref{assump:localgap} for the spectral gap stability argument. The red vertices comprise $\partial \Lambda_2^{(1)}$, and $\mr{\Lambda}_2^{(1)}$ corresponds to the black vertices and edges between them. We suppress the dependence of $\tilde{x}\in\tilde{\Gamma}_0$ for simplicity. }\label{fig:hex_ring}
\end{figure}
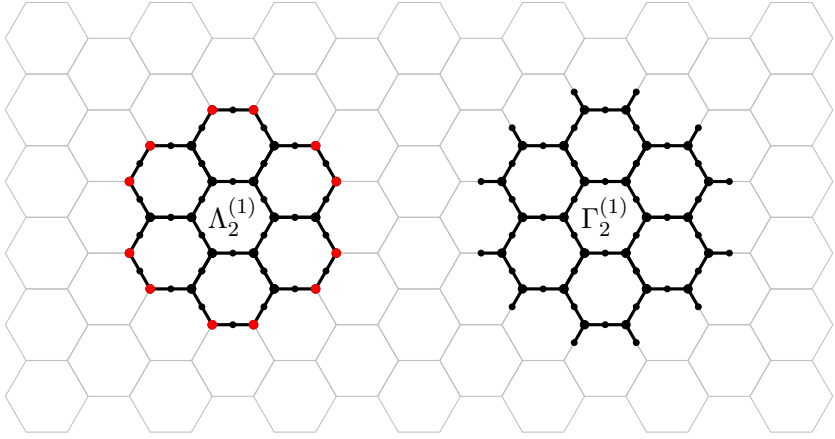

To illustrate this notation, let us introduce a family of finite volumes that will be frequently considered in this work. Denote by $\tilde{\Gamma}^{(0)}$ the dual lattice of $\Gamma^{(0)}$ (i.e. the triangular lattice which has a site at the center of every hexagon from $\Gamma^{(0)}$), and let $\tilde{D}$ be the graph distance on $\tilde{\Gamma}^{(0)}$. Let $\Lambda_{1}^{(d)}(\tilde{x})\subset\Gamma^{(d)}$ denote the $d$-decorated hexagon centered at any $\tilde{x}\in \tilde{\Gamma}^{(0)}$, and define
\begin{equation}\label{eq:Lambda_n}
	\Lambda_n^{(d)}(\tilde{x}) : = \bigcup_{\tilde{y}\in b_{n-1}(\tilde{x})} \Lambda_{1}^{(d)}(\tilde{y})\subseteq \Gamma^{(d)}
\end{equation}
for any $n\in\bN$ where $ b_{n}(\tilde{x})=\{\tilde{y}\in\tilde{\Gamma}^{(0)} : \tilde{D}(\tilde{x},\tilde{y})\leq n\}$. Visually, these volumes are formed by the union of $n-1$ concentric hexagon rings around $\Lambda_{1}^{(d)}(\tilde{x})$, see Figure~\ref{fig:hex_ring}. By construction, every vertex $v \in \partial \Lambda_n^{(d)}(\tilde x)$ is a vertex of the undecorated lattice that has degree two in $\Lambda_n^{(d)}(\tilde x)$, meaning that two of the three edges incident to $v$ belong to $\Lambda_n^{(d)}(\tilde x)$. Therefore, for any such $v$ there is a unique edge incident to $v$ not belonging to $\Lambda_n^{(d)}(\tilde x)$. This observation will be used in many arguments throughout this work. A simple counting argument shows that $|\partial \Lambda_n^{(d)}(\tilde x)|=6n$, see Lemma~\ref{lem:boundary-length}.

In their seminal work, Affleck, Kennedy, Lieb and Tasaki introduced their isotropic, antiferrmagnetic spin-$1$ chain and showed it satisfied the three properties of the Haldane phase. They also introduced several generalizations of their model, including the one considered in this work. This is defined by taking a $2s_x+1$-dimensional onsite Hilbert space  $\mathfrak{H}_x$ at every vertex $x\in \Gamma^{(d)}$ where $s_x := \deg(x)/2$. Then, for any finite volume $\Lambda \subseteq \Gamma^{(d)}$, the Hilbert space of states is given by $\mathfrak{H}_{\Lambda} :=\bigotimes_{x\in\Lambda}\mathfrak{H}_x$ and the algebra of observables is $\caA_{\Lambda}:= B(\mathfrak{H}_{\Lambda})$. The spin-$s_x$ irreducible representation of $\mathfrak{su}(2)$ acts on $\mathfrak{H}_x$, and for any finite volume $\Lambda\subseteq \Gamma^{(d)}$ the \emph{$d$-decorated AKLT Hamiltonian} acting on $\mathfrak{H}_{\Lambda}$ is
 \begin{equation}\label{eq:local_ham}
 H_{\Lambda} = \sum_{(x,y)\in \Lambda} P_{(x,y)}
 \end{equation}
 where $P_{(x,y)}\in\caA_{\{x,y\}}$ is the orthogonal projection onto the subspace of maximal spin $s_x+s_y$ from $\mathfrak{H}_x\otimes\mathfrak{H}_y$. Note that for the decorated hexagonal lattice, $s_x+s_y \in \{2,5/2\}$ for all edges $(x,y)\in\Gamma^{(d)}$. In Sections~\ref{sec:gss}-\ref{sec:main} we will use the Weyl representation of $\mathfrak{su}(2)$ to explicitly realize this model. 
 
 The Hamiltonian in \eqref{eq:local_ham} implicitly uses $\cA_{\Lambda_1}\subseteq \caA_{\Lambda_2}$ for any $\Lambda_1\subseteq \Lambda_2$ where one identifies  $ \cA_{\Lambda_1}\ni A \mapsto A\otimes\1_{\Lambda_2\setminus\Lambda_1}\in \cA_{\Lambda_2}$. As such, the support of any $A\in\cA_\Lambda$ is defined to be the smallest set $X$ such that $A$ acts as the identity on $\Lambda\setminus X$. Moreover, the algebra of local observables $\caA_{\Gamma^{(d)}}^{\rm loc} := \bigcup_{|\Lambda|<\infty} \cA_{\Lambda}$ is well-defined via the inductive limit induced by this identification. The $C^{*}$-algebra of quasi-local observables is then defined as the norm closure
 \[\caA_{\Gamma^{(d)}}:= \overline{\caA_{\Gamma^{(d)}}^{\rm loc}}^{\|\cdot\|}.\] We note that while the algebras $\caA_{\Gamma^{(d_1)}} \cong \caA_{\Gamma^{(d_2)}}$ are isomorphic when $d_1\geq d_2>0$ \cite{Glimm1960}, since we will often compare operators and states associated with the decorated and undecorated models, the decoration will be kept in the notation for clarity.
 
 In recent years, a number of general results have appeared for establishing spectral gap stability for ground states of quantum lattice models associated to finite-range, frustration-free gapped models. In addition to being finite-range, the decorated AKLT model on the decorated hexagonal lattice is both frustration-free and uniformly gapped, and so it is a natural candidate for applying these results. Frustration-freeness is the property that the ground states of any local Hamiltonian simultaneously minimize the energy of every interaction term. In the case of the decorated AKLT model, the ground state space is given by the (nontrivial) kernel of the Hamiltonian, and so the pure ground states are given by linear functionals $\varphi:\mathcal{A}_\Lambda\to\bC$ of the form
 \[
 \varphi(A) = \frac{\ip{\psi, A\psi}}{\|\psi\|^2}, \qquad 0\neq \psi \in \ker(H_{\Lambda}),
 \]
where $A\in \caA_{\Lambda}$ is any bounded linear operator. As the interaction terms are non-negative, these states necessarily satisfy $\varphi(P_{(x,y)}) =0$ for all edges $(x,y)\in\Lambda$.

The objects of interest for proving spectral gap stability are the infinite volume, frustration-free ground states. In the case of a model with a frustration-free interaction, a state on the quasi-local algebra is called \emph{frustration-free} if the expectation of any interaction term is zero. As is common for AKLT models on graphs with sufficiently small degree (and will be proved for $d\geq 5$ in Theorem~\ref{thm:indistinguishable}), there is a unique frustration-free state $\omega^{(d)}:\caA_{\Gamma^{(d)}} \to \bC$ for each $d$-decorated hexagonal model, i.e.
\begin{equation} \label{ffstate}
\omega^{(d)}(P_{(x,y)}) = 0 \qquad \forall\;(x,y)\in \Gamma^{(d)}.
\end{equation}

An elementary calculation shows that a state $\omega:\caA_{\Gamma^{(d)}}\to\bC$ is frustration-free if and only if it is the weak-* limit of finite-volume ground states, that is, if and only if there is an increasing and absorbing sequence of finite volumes $\Lambda_n\subseteq \Lambda_{n+1}$ so that $\bigcup_{n\geq 1}\Lambda_n = \Gamma^{(d)}$, and associated ground states $\varphi_n:\caA_{\Lambda_n} \to \bC$ with
\[
\omega(A) = \lim_{n\to\infty}\varphi_n(A), \quad \forall\, A\in\caA_{\Gamma^{(d)}}^{\rm loc}.
\]
( We note that the fact that any such limit is a frustration-free is trivial, while the reverse implication can be proven by showing that any frustration-free ground state is the weak-* limit of the finite volume ground states obtained by restricting itself to the local algebras associated to an increasing and absorbing sequence).

As mentioned in the introduction, we need to show that the AKLT model satisfies LTQO in order to prove spectral gap stability.  This will be a consequence of showing that the convergence of any sequence of finite volume ground states to the frustration-free ground state is sufficiently fast. This is the content of the indistinguishability result, Theorem~\ref{thm:indistinguishable} below. This result is stated with respect to the sequence of finite volumes $\Lambda_n^{(d)}:=\Lambda_n^{(d)}(\tilde{0})$ associated to some fixed point $\tilde{0}\in \tilde{\Gamma}^{(0)}$.

\begin{theorem}[Ground State Indistinguishability]\label{thm:indistinguishable} For the decorated AKLT model with $d\geq 5$, there is a frustration-free state $\omega^{(d)}:\caA_{\Gamma^{(d)}}\to\bC$ so that for any normalized $\psi_n\in \ker(H_{\Lambda_n^{(d)}})$ and observable $A\in\cA_{\mr{\Lambda}_k^{(d)}}$ with $1\leq k < n$,
\begin{equation}\label{eq:decay_bound}
|\braket{\psi_n}{A\psi_n}-\omega^{(d)}(A)| \leq 2F_\alpha(n,k)e^{F_\alpha(n,k)}\norm{A}
\end{equation}
where $F_\alpha(n,k)=102ke^{-2\alpha(n-k)}$ and, with respect to $f(x) = \frac{x+1-\sqrt{x^2+1}}{x}$, 
\begin{equation}
\alpha :=d\ln(3)-\ln\left(2\sqrt[5]{9}/f\left(4e\sqrt[5]{9}\right)\right)-4/e-.03.
\end{equation}
\end{theorem}
This theorem will be proven at the end of Section~\ref{sec:indistinguishability-proof}.
The assumption $d\geq 5$ guarantees $\alpha>0$ and so $F_\alpha(n,k)\to0$ as $n\to\infty$ for fixed $k$. Since the support of any local observable is contained in some $\mr{\Lambda}_k^{(d)}$, it follows that $\omega^{(d)}$ is the unique frustration-free state of the model. %
A unique infinite volume frustration-free state also exists for $d<5$, see, e.g. \cite{Kennedy1988}. However, the approach used to obtain the explicit constant and decay function from \eqref{eq:decay_bound} requires $d$ sufficiently large.

We specifically apply the stability result \cite[Theorem 2.8]{NSY:2021}, which is formulated in the infinite volume setting. An informal statement of this result will be given in Section~\ref{sec:stability} after the required assumptions are stated. In the infinite volume setting, the ground state and spectral gap properties are formulated in terms of the generator of the infinite volume dynamics. For any $d\geq 5$, \cite[Theorem 2.1]{BRLLNY} proved that the the AKLT model on the $d$-decorated hexagonal lattice has a uniform spectral gap $\gamma^{(d)}>0$ above its ground state energy. This in turn implies that the frustration-free state $\omega^{(d)}$ is a \emph{gapped ground state} of the infinite system dynamics generated by the closed derivation $\delta^{(d)}$ associated with the model, and its gap is at least $\gamma^{(d)}$. The domain of the generator contains $\cA_{\Gamma^{(d)}}^{\rm loc}\subseteq \dom(\delta^{(d)})$ as a core, and it is defined on such observables by
\[
\delta^{(d)}(A) := \lim_{\Lambda\uparrow \Gamma^{(d)}}[H_\Lambda, A]=\sum_{\substack{ (x,y)\in \Gamma^{(d)}}}[P_{(x,y)},A], \quad A\in \cA_{\Gamma^{(d)}}^{\rm loc}.
\]
In the GNS representation of $\omega^{(d)}$, this derivation is implemented by a self-adjoint operator called the GNS Hamiltonian, and the uniform gap $\gamma^{(d)}$ is a lower bound on the spectral gap above the ground state energy of this operator. As the frustration-free state is unique, a straightforward calculation shows that $\gamma^{(d)}$ being a lower bound on the gap of the GNS Hamiltonian is equivalent to the bound
\begin{equation}\label{eq:d_gap}
\omega^{(d)}(A^*\delta^{(d)}(A)) \geq \gamma^{(d)}\omega^{(d)}(A^*A) \quad \forall A\in\cA_\Gamma^{\rm loc} \quad \text{such that} \quad \omega^{(d)}(A) = 0.
\end{equation}
The second main result in this work proves that the gap in \eqref{eq:d_gap} is stable under sufficiently short range perturbations.

Explicitly, we consider perturbations defined by interactions $\Phi:\tilde{\Gamma}^{(0)}\times\bN \to \caA_{\Gamma^{(d)}}^{\rm loc}$ such that $\Phi(\tilde{x},n)^*=\Phi(\tilde{x},n)\in\cA_{\mr{\Lambda}_n^{(d)}(\tilde{x})}$, for which there are positive constants $a,\|\Phi\|>0$ and $0<\theta\leq 1$ so that
\begin{equation}\label{pert_decay}
 \|\Phi(\tilde{x},n)\| \leq \|\Phi\| e^{-an^\theta} \qquad \forall\, \tilde{x}\in\tilde{\Gamma}^{(0)},\; n\geq 1.
\end{equation}
This decay assumption guarantees the existence of an infinite-volume dynamics with a Lieb-Robinson bound for the perturbation, see, e.g., \cite{NOS,NSY} and Propostion~\ref{prop:LR_bound} in the appendix. Similar to above, the generator of this dynamics is a closed derivation $\delta^{\Phi}$ with $\cA_{\Gamma^{(d)}}^{\rm loc}\subseteq \dom(\delta^{\Phi})$ where the image of any local observable is given by the absolutely summable series
\begin{equation}\label{eq:pert_generator}
\delta^\Phi(A) := \sum_{(\tilde{x},n)\in\tilde{\Gamma}^{(0)}\times \bN}[\Phi(\tilde{x},n),A], \qquad A\in\cA_{\Gamma^{(d)}}^{\rm loc}.
\end{equation}

Stability is proven for the infinite volume system whose dynamics is generated by
\begin{equation}\label{perturbed_derivation}
\delta_s := \delta^{(d)} + s\delta^{\Phi}, \quad 0 \leq s \leq 1,
\end{equation}
which has $\dom(\delta^{(d)})\cap \dom(\delta^\Phi)$ as a core. The result below states that if $d\geq 5$, then for every fixed $0<\gamma<\gamma^{(d)}$ there is an associated $s_\gamma>0$ so that for each $0\leq s \leq s_\gamma$, the state
\begin{equation} \label{perturbed_gs}
\omega_s^{(d)}:=\omega^{(d)}\circ\alpha_s^\gamma
\end{equation} 
is a gapped ground state of $\delta^{(s)}$ with gap lower bounded by $\gamma$. Here, $\alpha_s^\gamma:\caA_{\Gamma^{(d)}}\to \caA_{\Gamma^{(d)}}$ is the quasi-adiabatic continuation (also known as the spectral flow) induced by the perturbed system first introduced by Hastings and Wen in \cite{Hastings2005}. This family of automorphisms, $\{\alpha_s^\gamma:s\in[0,1]\}$, has been integral in the study of gapped ground state phases. For a rigorous definition of the quasi-adiabatic continuation as well as proofs of its key properties, see \cite{BMNS, NSY}.

\begin{theorem}[Spectral Gap Stability]\label{thm:stability} Fix $d\geq 5$ and let $\gamma^{(d)}$ be the uniform gap of the AKLT model on $\Gamma^{(d)}$. For each $0<\gamma<\gamma^{(d)}$ and interaction $\Phi:\tilde{\Gamma}^{(0)}\times\bN \to \caA_{\Gamma^{(d)}}^{\rm loc}$ such that \eqref{pert_decay} holds, there exists $s_\gamma>0$ so that for all $0\leq s \leq s_\gamma$
	\begin{equation}\label{stable_gap}
	\omega_s^{(d)}(A^*\delta_s(A)) \geq \gamma \omega_s^{(d)}(A^*A) \quad \forall A\in\caA_{\Gamma^{(d)}}^{\rm loc} \;\; \text{s.t.} \;\; \omega_s^{(d)}(A)=0.
	\end{equation}
	Here, $\delta_s$ and $\omega_s^{(d)}$ are as in \eqref{perturbed_derivation}-\eqref{perturbed_gs}, respectively.
\end{theorem}
The proof of this theorem is the focus of Section~\ref{sec:stability}.

\section{Stability of the frustration-free ground state}\label{sec:stability}

\begin{figure}
	\begin{center}
		\begin{tikzpicture}
			\def\n{.35};
			\newcommand{\hexcoord}[2]
			{[shift=(60:#1),shift=(120:#1),shift=(0:#2),shift=(-60:#2),shift=(0:#2),shift=(60:#2)]}
			\foreach \x in {-1,...,11}
			\foreach \y in {1,...,9}{
				\draw[color=black!25]\hexcoord{\x*\n}{\y*\n}
				(0:\n)--(60:\n)--(120:\n)--(180:\n)--(-120:\n)--(-60:\n)--cycle;
				\draw[color=black!25,shift=(-60:\n), shift=(0:\n)]\hexcoord{\x*\n}{\y*\n}
				(0:\n)--(60:\n)--(120:\n)--(180:\n)--(-120:\n)--(-60:\n)--cycle;
			}
			
			\foreach \x in {0,60,120,180,240,300,360}
			{\draw[shift=(\x+120:\n), shift=(\x+60:\n), very thick, color=blue]\hexcoord{5*\n}{5*\n}
				(\x+0:\n)--(\x+60:\n)--(\x+120:\n)--(\x+180:\n);}
			\foreach \m in {-1,1}
			\foreach \l in {-1,1}
			\foreach \x in {0,60,120,180,240,300,360}
			{\draw[shift=(\x+120:\n), shift=(\x+60:\n), very thick, color=blue]\hexcoord{5*\n+\m*2*\n}{5*\n+\l*2*\n}
				(\x+0:\n)--(\x+60:\n)--(\x+120:\n)--(\x+180:\n);}
			\foreach \m in {-2,2}
			\foreach \x in {0,60,120,180,240,300,360}
			{\draw[shift=(\x+120:\n), shift=(\x+60:\n), very thick, color=blue]\hexcoord{5*\n+\m*2*\n}{5*\n}
				(\x+0:\n)--(\x+60:\n)--(\x+120:\n)--(\x+180:\n);}
			
			\foreach \m in {-1}
			\foreach \l in {-1,1} 
			\foreach \x in {30, 150}
			{\draw[violet!60, dashed, thick]\hexcoord{5*\n+\m*2*\n}{5*\n+\l*2*\n} (\x:\n*8)--(\x+180:\n*7);}
			\foreach \m in {-2}
			\foreach \x in {30, 150}
			{\draw[violet!60, dashed, thick]\hexcoord{5*\n+\m*2*\n}{5*\n} (0:0)--(\x+180:\n*8);}
			\foreach \x in {30, 90}
			{\draw[violet!60, dashed, thick]\hexcoord{7*\n}{3*\n} (0:0)--(\x+120:\n*8);
				\draw[violet!60, dashed, thick]\hexcoord{7*\n}{7*\n} (0:0)--(\x-60:\n*8);		
			}
			\foreach \x in {30, 150}
			{\draw[violet!60, dashed, thick]\hexcoord{9*\n}{5*\n} (0:0)--(\x:\n*8);}
			
			\draw[color=red, ultra thick]\hexcoord{5*\n}{5*\n}(0:0)--(30:6.7*\n);
			\draw[color=red,ultra thick]\hexcoord{5*\n}{5*\n}(0:0)--(150:6.7*\n);
			\draw[color=red,ultra thick, dashed]\hexcoord{7*\n}{3*\n}(0:0)--(30:7*\n);
			\draw[color=red,ultra thick, dashed]\hexcoord{7*\n}{7*\n}(0:0)--(150:7*\n);
			\draw[color=black,ultra thick, ->]\hexcoord{5*\n}{5*\n}(0:0)node[below] {$\mathbf{\tilde{0}}$}--(30:1.7*\n);
			\draw[color=black,ultra thick, ->]\hexcoord{5*\n}{5*\n}(0:0)--(150:1.7*\n);
		\end{tikzpicture}
		
	\end{center}
	\caption{Illustration of the separating partition. The part $\caT_2^{\tilde{0}}$ is the set of the dual lattice points where two dotted lines intersect. The index set $\caI_2\subseteq \tilde{\Gamma}^{(0)}$ is set of points contained in the fundamental cell outlined in red.}
	\label{fig:separating_partition}
\end{figure}
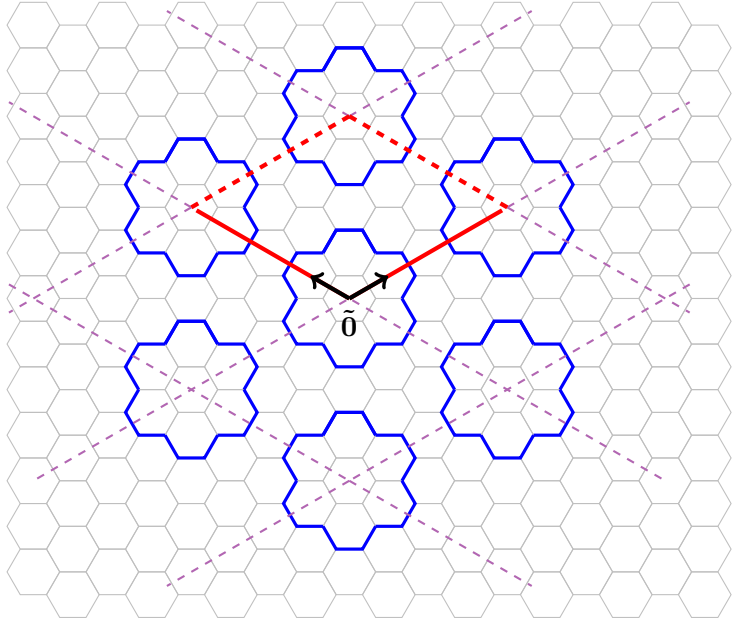

We now turn to proving Theorem~\ref{thm:stability} under the assumption that Theorem~\ref{thm:indistinguishable} holds. This is achieved by applying the spectral gap stability result \cite[Theorem~2.8]{NSY:2021}. The perturbations considered in \cite{NSY:2021} were supported on the balls of the lattice with respect to some nice metric, e.g., the graph distance on $\Gamma^{(d)}$. However, the choice to prove stability in that context was only to ensure certain key quantities were summable. This was a consequence of the fact that (1) the perturbation terms could be indexed by two sets: the vertices of the lattice and $\bN$, and (2) the number of sites contained in the support of any perturbation term grew at most like a polynomial in the radius of the ball. With mild changes to the notation in the proof of \cite{NSY:2021}, perturbations supported on other families of finite volumes with similar properties can also be used in the stability argument, so long as its dynamics satisfies a Lieb Robinson bound with stretched exponential decay. Here, the perturbation terms are supported on volumes which are indexed by $\tilde{\Gamma}^{(0)}$ and $\bN$, see \eqref{eq:Lambda_n} and \eqref{pert_decay}, which we will show satisfy the lattice regularity condition (Assumption~\ref{assump:regularity}) below. Hence, the stability argument can be adapted to this setting. The main adaptations are stated below in Assumptions~\ref{assump:regularity}-\ref{assump:LTQO}. In Appendix~\ref{appendix:hexagon_lattice}, we prove the necessary conditions which implies the perturbations satisfy a Lieb-Robinson bound for a function with stretched-exponential decay. The other minor changes one needs to make in \cite{NSY:2021} are also outlined in the appendix. 

We now state the main assumptions for \cite[Theorem 2.8]{NSY:2021} in the context considered here. The first is that the size of the volumes supporting the perturbation terms do not grow too quickly.

\begin{assumption}[Lattice Regularity]\label{assump:regularity} There are $\kappa, \nu>0$ such that for all $n\in\bN$ and $\tilde{x}\in\tilde{\Gamma}^{(0)}$:
	\begin{equation}
		|\Lambda_n^{(d)}(\tilde{x})| \leq \kappa n^\nu
	\end{equation}
\end{assumption}

In addition, there are three assumptions related to the unperturbed model. The first is a gap condition on the local Hamiltonians supported on volumes that are comparable to those that support the perturbation terms.

\begin{assumption}[Local Gap]\label{assump:localgap} There exists $\gamma>0$ and a family of finite volumes 
	\[\left\{\Gamma_n^{(d)}(\tilde{x}): \Lambda_n^{(d)}(\tilde{x})\subseteq \Gamma_n^{(d)}(\tilde{x})\; \forall \,n\in\bN, \, \tilde{x}\in\tilde{\Gamma}^{(0)}\right\}\] 
	so that $\inf_{\tilde{x},n}{\rm gap}(H_{\Gamma_n^{(d)}(\tilde{x})})\geq \gamma.$
\end{assumption}

We note that while the positive uniform gap in Assumption~\ref{assump:localgap} is sufficient for stability, a less stringent local gap condition is needed to prove stability as long as the infinite-volume frustration-free ground state is gapped, see \cite[Assumption~2.2]{NSY:2021}. 

Recall that the spectral gap above the ground state of any finite volume Hamiltonian, $H_\Lambda$, for the AKLT model on the decorated lattice is the difference between its ground state and first excited state energies, i.e.
\[
{\rm gap(H_\Lambda)} = E_\Lambda^1 - E_\Lambda^0, \quad E_\Lambda^0 = \min{\rm spec}(H_\Lambda), \quad E_\Lambda^1 = \min {\rm spec}(H_\Lambda)\setminus \{E_\Lambda^0\}.
\]
For the AKLT model on the decorated hexagonal lattice, the local gap condition will be an immediate consequence of \cite[Theorem 2.2]{BRLLNY}, which we review for the reader's convenience. For any spin-3/2 vertex $v\in \Gamma^{(0)}$ of the undecorated lattice, let $Y_v^{(d)}\subset \Gamma^{(d)}$ denote the subvolume of $3d+1$ sites consisting of $v$ and the three spin-1 chains of length $d$ emanating from $v$. Suppose that $S\subseteq \Gamma^{(0)}$ is any finite set of spin-3/2 vertices so that
\begin{equation} \label{eq:aklt_gap_volumes}
\Lambda = \bigcup_{v\in S} Y_v^{(d)}.
\end{equation}
Then, \cite[Theorem 2.2]{BRLLNY} states that there exists $\gamma^{(d)}>0$ so that for any $\Lambda$ as above,
\begin{equation}\label{eq:local_gap}
{\rm gap}(H_\Lambda) \geq \gamma^{(d)},
\end{equation}
see also the comments following \cite[Equation~2.1]{BRLLNY}.

The next assumption guarantees that for each $n\in\bN$, the collection of subvolumes $\{\Gamma_n^{(d)}(\tilde{x}): \tilde{x}\in\tilde{\Gamma}^{(0)}\}$ can be partitioned into polynomially many sets, each of which consists of a collection of subvolumes that are spatially disjoint.

\begin{assumption}[Separating Partition of Polynomial Growth]\label{assump:partition} For each $n\geq 1$, there exists a index set $\caI_n$ and a partition $\{\caT_n^m: m\in\caI_n\}$ of  $\tilde{\Gamma}^{(0)}$ indexed by $\caI_n$, so that the following conditions hold:
	\begin{enumerate}
		\item For every $m\in\caI_n$, if $\tilde{x},\tilde{y}\in\caT_n^m$ are distinct, then $\Gamma_n^{(d)}(\tilde{x}) \cap \Gamma_n^{(d)}(\tilde{y})= \emptyset$.
		\item There are constants $\kappa_0,\nu_0>0$ such that $|\caI_n| \leq \kappa_0n^{\nu_0}$ for all $n\geq 1$.
	\end{enumerate}
\end{assumption}

The separating partition condition along with the lattice regularity assumption are used to characterize how indistinguishable the finite volume ground states need to be in order to guarantee the spectral gap is stable. This is captured by the local topological quantum order (LTQO) assumption. This property is only ever applied to perturbation terms in the stability argument, and so the support of the observables in the assumption below only needs to match the support of the perturbation terms from \eqref{pert_decay}.

\begin{assumption}[Local Topological Quantum Order]\label{assump:LTQO} Let $\omega^{(d)}$ be the frustration-free ground state from Theorem~\ref{thm:indistinguishable}, and denote by $P_{n}^{(d)}(\tilde{x})$ the orthogonal projection onto the ground state space $\ker H_{\Lambda_n^{(d)}(\tilde{x})}$. There is a non-increasing function $G:[0,\infty)\to[0,\infty)$ satisfying
	\[
	\sum_{n\geq 1}n^{\nu_0+\nu/2}G(n)<\infty
	\]
such that for all $ n\geq 2k\geq 2$, $\tilde{x}\in\tilde{\Gamma}^{(0)}$ and $A\in\caA_{\mr{\Lambda}_k^{(d)}(\tilde{x})}$
\[
\|P_{n}^{(d)}(\tilde{x})AP_{n}^{(d)}(\tilde{x})-\omega^{(d)}(A)P_{n}^{(d)}(\tilde{x})\| \leq |\Lambda_k^{(d)}(\tilde{x})|G(n-k)\|A\|.
\]
\end{assumption}

While we have written the assumptions above in the context of the decorated AKLT model, generalizations of these criterion hold in more general contexts. Informally, \cite[Theorem~2.8]{NSY:2021} states the following: Suppose that $\omega$ is the unique frustration-free ground state associated to a quantum spin model defined by a finite-range interaction whose terms are uniformly bounded in norm. Assume that $\gamma_0>0$ is a lower bound on the spectral gap of $\omega$ in the sense that \eqref{eq:d_gap} holds for $\omega$ and $\gamma_0$. If this model satisfies Assumptions~\ref{assump:regularity}-\ref{assump:LTQO}, then for any $0<\gamma<\gamma_0$ and any perturbation decaying at least as fast as a stretched exponential as in \eqref{pert_decay}, there exists $s_\gamma$ so that for all $0\leq s \leq s_\gamma$, the state $\omega_s = \omega\circ\alpha_s^\gamma$ is a ground state of $\delta_s$ from \eqref{perturbed_derivation} and, moreover, this is a gapped ground state in the sense that \eqref{stable_gap} holds. 

Said more concisely, \cite[Theorem~2.8]{NSY:2021} states that if Assumptions~\ref{assump:regularity}-\ref{assump:LTQO} holds, then Theorem~\ref{thm:stability} holds for any perturbation satisfying \eqref{pert_decay}. Thus, one only needs to verify these assumptions. In most cases, the local gap and LTQO conditions are the most difficult assumptions to verify. The next result shows that the LTQO condition is an immediate consequence of Theorem~\ref{thm:indistinguishable}.

\begin{corollary}[LTQO]\label{cor:LTQO}
	Suppose $d \geq 5$ and  $ n\geq 2k\geq 2$. For any $A\in \caA_{\mr{\Lambda}_k^{(d)}(\tilde{x})}$ with $\tilde{x} \in \tilde{\Gamma}^{(0)}$
	\begin{align}
		\norm{ P_{n}^{(d)}(\tilde{x})AP_{n}^{(d)}(\tilde{x})-\omega^{(d)}(A)P_{n}^{(d)}(\tilde{x})} \leq  |\partial\Lambda_k^{(d)}(\tilde{x})|G_{\alpha}(n-k)\|A\|
	\end{align}
where $G_{\alpha}(r) = C_\alpha e^{-2\alpha r}$ with $C_\alpha = 68e^{51/\alpha e}$ and $\alpha$ as in Theorem~\ref{thm:indistinguishable} .
\end{corollary}

\begin{proof} Since the AKLT model is invariant under any translation of the dual lattice $\tilde{\Gamma}^{(0)}$, it is sufficient to consider the ground state projection $P_n^{(d)}$ associated to $\Lambda_n^{(d)}:=\Lambda_n^{(d)}(\tilde{0})$ for some fixed $\tilde{0}\in\Gamma^{(d)}$. If $A\in\caA_{\mr{\Lambda}_k^{(d)}}$ is self-adjoint, then by Theorem~\ref{thm:indistinguishable}
	\begin{align}\label{ltqo_bound}
		\norm{ P_n^{(d)} A P_n^{(d)} - \omega^{(d)}(A) P_n^{(d)}} = \sup _{ \substack{ \Psi \in \ker(H_{\Lambda_n^{(d)}}) : \\ \norm{\Psi}=1}} | \ip{\Psi , A \Psi } - \omega^{(d)}(A)| \leq 2F_\alpha(n,k)e^{F_\alpha(n,k)}\|A\|.
	\end{align}
	If $A$ is not self-adjoint, we can decompose it as $A = B + iC$, where $B$ and $C$ are self-adjoint and $\norm{B}, \norm{C} \le \norm{A}$. Using the triangle inequality and applying \eqref{ltqo_bound} to $B$ and $C$ independently, we obtain
	\begin{equation}\label{ltqo_bound2}
\norm{ P_n^{(d)} A P_n^{(d)} - \omega (A) P_n^{(d)}}	\le 4 F_\alpha(n,k)e^{F_\alpha(n,k)}\|A\|.
	\end{equation}
 The boundary $\partial\Lambda_k^{(d)}$ contains $6k$ sites by Proposition~\ref{lem:boundary-length}. Therefore, $F_\alpha(n,k) = 17e^{-2\alpha(n-k)}|\partial\Lambda_k^{(d)}|.$ Alternatively, since $k\leq n/2$
\[
F_\alpha(n,k) \leq 51ne^{-\alpha n} \leq  \frac{51}{\alpha e}
\]
as $re^{-r} \leq 1/e$ for all $r\geq 0$. Inserting these into \eqref{ltqo_bound2} produces the result.

\end{proof}

We now prove Theorem~\ref{thm:stability} under the assumption that Theorem~\ref{thm:indistinguishable} holds.

\begin{proof}[Proof of Theorem~\ref{thm:stability}] The simple counting argument from Proposition~\ref{lem:boundary-length} shows that the lattice regularity condition is satisfied since  $|\Lambda_n^{(d)}(\tilde{x})|\leq 3(3d+2)n^2$, and Assumption~\ref{assump:LTQO} holds by Corollary~\ref{cor:LTQO}. For the local gap condition, for any $\tilde{x}\in \tilde{\Gamma}^{(0)}$ and $n\in \bN$, let 
\begin{equation}\label{eq:lg_volumes}
\Gamma_{n}^{(d)}(\tilde{x}) := \bigcup_{v\in\Lambda_n^{(0)}(\tilde{x})} Y_v^{(d)},
\end{equation}
where we recall that $Y_v^{(d)}$ is as defined above \eqref{eq:aklt_gap_volumes}. Note that this is the union of $\Lambda_n^{(d)}(\tilde{x})$ and all decorated spin-1 chains emanating from a boundary site $v\in\partial\Lambda_n^{(d)}(\tilde{x})$,  see Figure~\ref{fig:hex_ring}. Hence, 
\[
\inf_{\tilde{x},n}{\rm gap}(H_{\Gamma_{n}^{(d)}(\tilde{x})})\geq \gamma^{(d)}>0
\]
as desired by \eqref{eq:local_gap}. 

It is left to verify Assumption~\ref{assump:partition}. To define the separating partition, let $\tilde{0}\in\tilde{\Gamma}^{(0)}$ denote some fixed site of the dual lattice, and let $v_\pm = (\pm\frac{\sqrt{3}}{2},\frac{1}{2})\in\bR^2$ be the two dual lattice vectors as in Figure~\ref{fig:separating_partition}. Then, \[\tilde{\Gamma}^{(0)}=\{\tilde{x}=\tilde{0}+kv_++\ell v_- : k,l\in\bZ\}\] 
and $\tilde{D}(\tilde{x},\tilde{0}) = |k|+|\ell|$. It is easy to verify that $|\caI_n|=4n^2$ if one chooses the index set for the $n$-th partition to be
\[
\caI_n = \{\tilde{m}=\tilde{0}+kv_+ + \ell v_-: 0\leq k,\ell \leq 2n-1\}\subseteq \tilde{\Gamma}^{(0)}.
\]
With this choice, the $n$-th partition part indexed by $\tilde{m}\in\caI_n$  can be taken as
\[
\caT_n^{\tilde{m}} = \{\tilde{m}+2n(kv_+ + \ell v_-) : k,\ell\in\bZ\}\subseteq \tilde{\Gamma}^{(0)}.
\]
This satisfies $\tilde{D}(\tilde{x},\tilde{y})\geq 2n$ for any two distinct $\tilde{x},\tilde{y}\in \caT_n^{\tilde{m}}$, which is the minimal distance needed to guarantee that $\Gamma_n^{(d)}(\tilde{x})\cap \Gamma_n^{(d)}(\tilde{y}) = \emptyset$, see Figure~\ref{fig:separating_partition}.

Thus, all of the stability assumptions hold and so by \cite[Theorem 2.8]{NSY:2021}, for any $0<\gamma<\gamma^{(d)}$ there is $s_\gamma>0$ so that for all $0\leq s \leq s_\gamma$ and any local observable such that $\omega_s^{(d)}(A)=0$,
\[
\omega_s^{(d)}(A^*\delta_s(A)) \geq \gamma \omega_s^{(d)}(A^*A).
\]
\end{proof}

\section{Characterizations of ground states}\label{sec:gss}

The remainder of this work focuses on proving Theorem~\ref{thm:indistinguishable}.  As such, from now on we only consider the fixed sequence $\Lambda_N^{(d)}:=\Lambda_N^{(d)}(\tilde{0})$ defined as in \eqref{eq:Lambda_n}. We emphasize once again that all vertices $v\in \Lambda_N^{(d)}$ satisfy $\deg(v)\in\{2,3\}$ and, in particular, all boundary vertices have degree three, and two of their edges are contained in $\Lambda_N^{(d)}$. To further simplify notation, set
\begin{equation}\label{notation}
\mf{H}_N^{(d)} : = \mf{H}_{\Lambda_N^{(d)}}, \quad H_N^{(d)}: = H_{\Lambda_N^{(d)}}, \quad \cA_N^{{(d)}}:= \cA_{\Lambda_N^{(d)}},  \quad \mr{\cA}_N^{{(d)}}:= \cA_{\mr{\Lambda}_N^{(d)}}.
\end{equation}
 Since we will frequently consider subgraphs of both $\Gamma^{(d)}$ and $\Lambda_N^{(d)}$,  we denote their respective sets of vertices and bonds by
 \begin{equation}\label{eq:graph_defs}
 \Gamma^{(d)} \equiv (\caV^{(d)}, \caB^{(d)}), \quad \Lambda_N^{(d)} \equiv (\caV_N^{(d)}, \caB_N^{(d)}).
\end{equation}

There are two goals of this section. The first is to give a nice description for the ground state space of $H_N^{(d)}$. This is achieved in Section~\ref{sec:hamiltonian} by considering the Weyl representation of $\su$ acting on a Hilbert space of homogeneous polynomials. Each ground state can be uniquely described by a polynomial supported on the boundary $\partial\Lambda_N^{(d)}$. As we are interested in calculating the ground state expectation of observables supported sufficiently far away from this boundary, a finite volume ``bulk-boundary map'' will be identified that can be used to calculate the expected value of any such observable in the ground state associated with any fixed boundary polynomial. A  ``bulk state'' will also be defined which will be used to prove Theorem~\ref{thm:indistinguishable} in Section~\ref{sec:main}. This will be done by showing that the bulk state well-approximates each of the finite volume ground states and, moreover, converges strongly to the unique infinite volume ground state. The second goal of this section is to rewrite these maps in terms of hard core polymer representations. The graphs and weights used for this representation are introduced in Section~\ref{sec:polymers_weights}, and the final expressions are proved in Section~\ref{sec:loop}. A lemma producing an initial comparsion between the bulk-boundary map and the bulk state is then proved in Section~\ref{sec:comparison}, from which we will obtain the indistinguishability bound in Section~\ref{sec:main}.

\subsection{The ground states and bulk state of the decorated AKLT Hamiltonian}\label{sec:hamiltonian}
We follow the construction in \cite{Kennedy1988}, and use the Weyl representation of the Lie algebra $\mf{su}(2)$ acting on polynomials in two variables to explicitly realize the AKLT mode on the decorated lattice. For the convenience of the reader, we review the relevant setup and ground state description from this work. As such, consider the Hilbert space of complex homogeneous polynomials of degree $m$
\begin{equation}
	\begin{split}
		\mc{H}^{(m)} := \set{ \sum_{k=0}^{m} \la_{k} u^k v^{m-k} \, : \, \lambda_k\in\bC  } \subset \cc [u,v],
	\end{split}
\end{equation} 
where the inner product is taken so that the monomial basis is orthogonal. Concretely, using the change of variables
\begin{align}\label{eq:change-of-coordinates}
	u= \exp(i \phi/2)\cos(\theta/2), \quad v=\exp(-i\phi/2) \sin(\theta/2),
\end{align} 
and given any pair $\Phi,\Psi\in\HH{m}$, the inner product is
\begin{align}
	\ip{\Phi, \Psi}& = \int d\Omega ~ \overline{\Phi(\theta, \phi)} \Psi(\theta,\phi) \label{eq:IP}\\
	d\Omega  = \frac{1}{4\pi} \sin & (\theta) d\phi d\theta, \;\; 0 \leq \phi < 2\pi, \; 0\leq \theta < \pi.\label{integrand_data}
\end{align}
In particular, this allows one to view each $\HH{m}$ as a subspace of $L^2(d\Omega).$

 The onsite Hilbert space for the decorated AKLT model is then $\mf{H}_x = \mc{H}^{(\deg(x))}$ for each vertex $x\in\Gamma^{(d)}$. Thus, $\mf{H}_\Lambda = \bigotimes _{x\in \Lambda}\mc{H}^{(\deg(x))}$ for any finite $\Lambda\subseteq \Gamma^{(d)}$, and the associated inner product is
\begin{align}
\ip{\Phi, \Psi}& = \int d\BOmega^{\Lambda} ~ \overline{\Phi(\mathbf{\theta, \phi})} \Psi(\mathbf{\theta, \phi}) , \qquad \forall \, \Phi,\Psi\in \mf{H}_\Lambda\,.
\end{align} 
Above, $d\mathbf{\Omega}^{\Lambda}$ is the product measure associated to $\{d\Omega_x : x\in \Lambda\}$ and  $\Phi(\mathbf{\theta, \phi})$ denotes the function resulting from appropriately applying the change of variables \eqref{eq:change-of-coordinates} independently to each pair of variables $u_x,v_x$  associated to any $x\in\Lambda$. For simplicity, we denote by $\theta = (\theta_x)_{x\in\Lambda}$, and $\phi=(\phi_x)_{x\in\Lambda}$.

The local AKLT Hamiltonian from \eqref{eq:local_ham} is represented on $\mf{H}_\Lambda$ using the Weyl representation. For each $m\geq 0$, this is the irreducible representation $\pi_m: \su \to B(\HH{m})$ given by
\begin{equation}\label{spin-mat}
	\begin{split}
		\pi_m(\sigma^3) = v\partial_v - u\partial_u, \hspace{2mm}
		\pi_m(\sigma^-) = u\partial_v ,\hspace{2mm}
		\pi_m(\sigma^+) = v\partial_u,
	\end{split}
\end{equation}
where $\sigma^3$ is the third Pauli matrix, and $\sigma^{\pm}$ are the usual lowering and raising operators. This is isomorphic to the spin-$m/2$ representation. For adjacent sites $x$ and $y$, with degrees $m_x$ and $m_y$ respectively, the subspace of $\mf{H}_x \otimes \mf{H}_y$ corresponding to the maximal spin $(m_x + m_y)/2$ is spanned by the states $(u_x\partial_{v_x} + u_y\partial_{v_y})^kv_x^{m_x}v_y^{m_y}$ where $0 \leq k \leq m_x+m_y$, as one can check by evaluating these states against the tensor product representation $\pi_{m_x}\otimes \pi_{m_y}$ of $\su$. The orthogonal projection onto this subspace then gives the AKLT interaction term $P_{(x,y)}\in B(\mf{H}_x \otimes \mf{H}_y)$.
In this representation, the ground state space $\ker{H_\Lambda}$ is characterized by boundary polynomials. Indeed, by a simple argument from \cite{Kennedy1988},
\begin{equation}\label{lem:ground-statespace}
	\begin{split}
		\ker{P_{(x,y)}} = \set{ f \in \mf{H}_x \otimes \mf{H}_y  : f = (v_xu_y - u_xv_y) g(u_x,v_x, u_y,v_y)},
	\end{split}
\end{equation}
where $g$ is a homogeneous polynomial of degree $\deg(x)-1$ in $u_x$ and $v_x$, and similarly in the $y$-variables. In words, the ground state requires that there be a singlet  $v_xu_y-u_xv_y$ across the bond $(x,y)$, but the remaining variables can form any homogeneous polynomial of the appropriate degree. By the frustration-free property, a ground state of any finite volume Hamiltonian $H_\Lambda$ must project all edges $(x,y)\in\Lambda$ into a singlet. Since $\mathfrak{H}_x = \cH^{(\deg(x))}$ and the polynomials over $\cc$ form a unique factorization domain, (\ref{lem:ground-statespace}) immediately implies the following description for the ground state space. 
\begin{theorem}[\hspace*{-3px}\cite{Kennedy1988,KK89}]\label{thm:gss}
	Let $d\geq 0$. For any finite $\Lambda\subseteq \Gamma^{(d)}$ the ground state space is given by
	\begin{equation}
		\ker(H_\Lambda) = \left\{\Psi= g\cdot \prod_{(i,j)\in\Lambda}(u_iv_j-v_iu_j)\in \mf{H}_\Lambda : g\in \mathfrak{H}_{\partial \Lambda}^{\rm gss}\right\}
	\end{equation}
where the set of all possible boundary polynomials is
\begin{equation}
\mathfrak{H}_{\partial \Lambda}^{\rm gss}=\bigotimes_{i\in\partial\Lambda} \caH^{(d_i)}, \quad d_i =  \deg(i)-\deg_\Lambda(i)\,.
\end{equation}
\end{theorem}
Note that $d_i$ is the number of edges connected to $i$ that are not contained in $\Lambda$. In the case of $\Lambda_N^{(d)}$, since $d_i=1$ for all $i\in\partial\Lambda_N^{(d)}$, see Figure~\ref{fig:hex_ring}, the space of boundary polynomials, $\mathfrak{H}_{\partial \Lambda_N^{(d)}}^{\rm gss}$, is spanned by all elements of the form
\[\prod_{i\in\partial\Lambda_N^{(d)}}\left(a_iu_i+(1-a_i)v_i\right),  \qquad a_i\in\{0,1\}.\]
Therefore, $\dim( \ker H_N^{(d)}) = 2^{|\partial\Lambda_N^{(d)}|}=2^{6N}$ by Proposition~\ref{lem:boundary-length}.

The matrix entries of an operator $A \in \cA_N^{(d)}$ can be conveniently described using the change of variables \eqref{eq:change-of-coordinates} by introducing the \emph{symbol} of $A$, denoted  $A(\BOmega)$. Let
\begin{equation}\label{eq:sphereical_coord}
	\Omega_x=(\sin{\theta_x}\cos{\phi_x},\sin{\theta_x}\sin{\phi_x},\cos{\theta_x} )
\end{equation}
be spherical coordinate associated to $x\in\Lambda_N^{(d)}$. Arovas, Auerbach and Haldane showed in \cite{PhysRevLett.60.531} that
\begin{align}\label{eqn:matrix_entries}
	\ip{ \eta, A \xi} = \int d\BOmega^{\Lambda_N^{(d)}} ~\overline{\eta(\theta,\phi)} \xi(\theta,\phi) A(\BOmega),  ~~~ \forall \eta, \xi \in \mf{H}_N^{(d)}.
\end{align}
where $A(\BOmega)$ is a continuous function of the angles $\theta_x,\phi_x$ associated to $x\in \supp{A}$. 

In general the symbol is not unique. However, a specific choice can be made by first defining it unambiguously for a basis of the onsite algebra $B(\caH^{(m)})$ and invoking linearity to define the symbol for a general $A\in B(\caH^{(m)})$. We require $\1(\Omega)=1$ so that the support condition stated after \eqref{eqn:matrix_entries} is satisfied. This is achieved by including $\1 = 1$ in the onsite basis and implementing the following procedure. First, use the commutation relation $[\partial_x, x]=1$ to rewrite each basis element as
\[
A = \sum_{k,l\in\bN_0}\sum_{-k\leq j\leq l}  a_{k,l,j} \partial_u^k \partial_v^l u^{k+j}v^{l-j}, \qquad  a_{k,l,j}\in\bC.
\]
Then, using $\braket{\Phi}{\partial_u^k\partial_v^lu^{k+j}v^{l-j}\Psi}=C_{k,l}\braket{u^kv^l\Phi}{u^{k+j}v^{l-j}\Psi}$ where $C_{k,l} = \frac{(m+l+k+1)!}{(m+1)!}$, define the symbol to be
\begin{equation}\label{eq:symbol}
A(\Omega):=\sum_{k,l\in\bN_0}\sum_{-k\leq j\leq l} C_{k,l}a_{k,l,j} \overline{ u^k v^l} u^{k+j}v^{l-j},
\end{equation}
which is to be understood using \eqref{eq:change-of-coordinates}. The formula extends to any $A\in\caA_{\Gamma^{(d)}}^{\rm loc}$ in the usual way: $\left(\bigotimes_{x}A_x\right)(\BOmega) := \prod_x A_x(\Omega_x)$ for a product of onsite observables, and then extended to any local operator by linearity. The convention $\1(\Omega)=1$ implies $AB(\BOmega) = A(\BOmega)B(\BOmega)$ if $A,B$ have disjoint support.

The matrix elements formula \eqref{eqn:matrix_entries} can also be used to calculate ground state expectations for any $\Psi(f)\in\ker H_N^{(d)}$ with boundary polynomial $f\in \mathfrak{H}_{\partial \Lambda_N^{(d)}}^{\rm gss}$ as
\begin{equation}\label{eqn:gs_matrix_entries}
	\braket{\Psi(f)}{A \Psi(f)} = \int d\BOmega^{\Lambda_N^{(d)}}\prod_{(i,j)\in\Lambda_N^{(d)}}|u_i v_j - v_i u_j|^2 |f|^2A(\BOmega).
\end{equation}
The change of variables (\ref{eq:change-of-coordinates}) can also be used to show $|u_i v_j - v_i u_j|^2 = \frac{1}{2} (1- \Omega_i \cdot \Omega_j).$ Thus, setting
\begin{equation}\label{def_rho}
	d\rho_{\Lambda_{N}^{(d)}}=\rho_{\Lambda_{N}^{(d)}} d\BOmega^{\Lambda_{N}^{(d)}}, \quad \rho_{\Lambda_N^{(d)}} = 2^{-|\cB_N^{(d)}|}\prod _{(i,j)\in \Lambda_{N}^{(d)}} (1-\Omega_i \cdot \Omega_j),
\end{equation}
the ground state expectation $\braket{\Psi(f)}{A \Psi(f)}$ can then be rewritten in terms of $|f|^2$ and a bulk-boundary map $\mr{\omega}_{N}(A;\partial \BOmega)$ that is independent of $f\in \mathfrak{H}_{\partial \Lambda_N^{(d)}}^{\rm gss}$ as follows.
\begin{lemma}[Bulk-boundary map]\label{lem:bulk-boundary_map}
	Fix $N\geq 2$ and let  $\Psi(f)\in \ker H_{N}^{(d)}$ be a nonzero ground state associated with a boundary polynomial $f\in \mathfrak{H}_{\partial \Lambda_N^{(d)}}^{\rm gss}$ as in Theorem~\ref{thm:gss}. Then, for any $K<N$,
	\begin{equation}\label{fv_gs_exp}
	\braket{\Psi(f)}{A \Psi(f)} = \int d\rho_{\Lambda_{N}^{(d)}}\,  |f|^2\, \mr{\omega}_N(A;\partial \BOmega), \qquad A\in \caA_K^{(d)}
	\end{equation}
where $\mr{\omega}_N(A;\partial \BOmega):=\mr{Z}_{N}(A;\partial \BOmega)/\mr{Z}_N(\partial \BOmega)$ is the function of the boundary variables  $\partial \BOmega = ( \Omega_x : x\in\partial\Lambda_N^{(d)})$ defined by
	\begin{equation}\label{bulk_map}
		\mr{Z}_{N}(A;\partial \BOmega) := \int d\BOmega^{\mr{\Lambda}_{N}^{(d)}} \rho_{\Lambda_{N}^{(d)}} A(\BOmega), \qquad \mr{Z}_N(\partial \BOmega):=\mr{Z}_N(\1;\partial \BOmega).
	\end{equation}
\end{lemma}

\begin{proof}
We first show that $\mr{\omega}_N(A;\partial \BOmega)$ is well-defined on all sets with positive measure. For any fixed choice of the boundary variables $\partial \BOmega$, the map $A \mapsto \mr{\omega}_N(A;\partial \BOmega)$ is a ground state of $H_{\mr{\Lambda}_N^{(d)}}$.
	To see this, fix the values of $v_i(\theta_i,\phi_i),u_i(\theta_i,\phi_i)$ for each $i\in\partial\Lambda_N^{(d)}$, and consider the function $g_{\partial \BOmega}$ defined by
	\begin{equation}\label{eq:interior_gs}
		g_{\partial \BOmega} = \prod_{\substack{(i,j)\in\Lambda_N^{(d)}: \\ i\in\partial\Lambda_N^{(d)}}}(u_iv_j-v_iu_j)\in \mathfrak{H}_{\partial \mr{\Lambda}_N^{(d)}}^{\rm gss}.  
	\end{equation}
Above, we observe that any site $j\in\Lambda_N^{(d)}$ that neighbors $i\in\partial\Lambda_N^{(d)}$ is necessarily an interior site for all $N\geq 2$, and so $g_{\partial \BOmega}$ is nonzero. By Theorem~\ref{thm:gss}, $\Psi(g_{\partial \BOmega})\in\ker(H_{\mr{\Lambda}_N^{(d)}})$, and, as a consequence,  $\mr{Z}_N(A;\partial \BOmega) = \braket{\Psi(g_{\partial \BOmega})}{A\Psi(g_{\partial \BOmega})}$.
	This implies that $\mr{Z}_N(\1) = \norm{\Psi(g_{\partial  \BOmega})}^2 \neq 0$, and $\mr{\omega}_N(A; \partial \BOmega)$ is a bounded, continuous function of the boundary variables for each $A$. Hence, \eqref{bulk_map} is well-defined, and so too is $\mr{\omega}_N(A;\partial \BOmega)$. As a consequence, \eqref{fv_gs_exp} follows immediately from the matrix element formula \eqref{eqn:gs_matrix_entries} since
	\begin{align*}
		\ip{ \Psi (f), A \Psi(f)}
		& =  \int d\BOmega^{\Lambda_N^{{(d)}}} |f|^2 \rho_{\Lambda_N^{(d)}} A(\BOmega) \\
		& = \int d\BOmega^{\partial\Lambda_N^{{(d)}}} |f|^2 \int d\BOmega^{\mr{\Lambda}_N^{(d)}} ~ \rho_{\Lambda_N^{(d)}}
		\left[ \frac{\int d\BOmega^{\mr{\Lambda}_N^{(d)}}\rho_{\Lambda_N^{(d)}} A(\BOmega)  }{\int d\BOmega^{\mr{\Lambda}_N^{(d)}} ~ \rho_{\Lambda_N^{(d)}}} \right].
	\end{align*}
\end{proof}

We now introduce the \emph{bulk state}, $\omega_N(A)$, which we show well-approximates $\braket{\Psi(f)}{A \Psi(f)}$ as in \eqref{fv_gs_exp} when $K<<N$. This is motivated from averaging the bulk-boundary function $\mr{Z}_N(A;\partial \BOmega)$ over the possible values of the boundary variables. Explicitly, $\omega_N(A):= Z_N(A)/Z_N$ where
\begin{align}\label{eq:unnormalized}
	Z_N(A) := \int d\rho_{\Lambda_{N}^{(d)}}~ A(\BOmega), \quad A\in \cA_N^{(d)}
\end{align}
and $Z_N := Z_N(\1)$. Note that, if $A \in \cA_K^{(d)}$ with $K<N$ one has that, indeed,
\[
  Z_N(A) = \int d\BOmega^{\partial \Lambda_N^{(d)}} \mr{Z}_N(A;\partial \BOmega).
\]

It is not immediately obvious from \eqref{eq:unnormalized} if $\omega_N$ is a ground state for $H_{N}^{(d)}$, or even if it is positive on all of $\caA_N^{(d)}$. However, it is a ground state of $H_K^{(d)}$ for all $K<N$. To see this, let us consider the AKLT model obtained by replacing the spin-3/2 at all boundary sites $x\in\partial \Lambda_N^{(d)}$ with a spin-1, where the nearest-neighbor interaction is still defined as the orthogonal projection onto the largest spin subspace between any pair of adjacent sites. The natural variation of Theorem~\ref{thm:gss} applies in this case, yielding a unique ground state given by $\Psi_N = \prod_{(i,j)\in\Lambda_N^{(d)}}(u_iv_j-v_iu_j)$.
We observe that, for $A \in \cA_K^{(d)}$ with $K<N$,
\[
  Z_N(A) = \braket{\Psi_N}{A\Psi_N}, \quad \quad Z_N = \norm{\Psi_N}^2.
\]
As this modified model does not change the spin or interaction terms for sites of $\Lambda_K^{(d)}$, $\omega_N$ is then a ground state of $H_K^{(d)}$ by frustration-freeness.

Notice that for any normalized $\Psi(f)\in\ker H_N^{(d)}$ and observable $A\in\cA_K^{(d)}$ with $K<N$ one
\begin{equation}\label{eq:bulk-groun-state-comparison}
	| \ip{\Psi(f), A \Psi(f)} - \omega_{N}(A) | = \left| \int d\rho_{\Lambda_{N}^{(d)}} \,|f|^2 \, \left[ \mr{\omega}_N(A;\partial \BOmega)-\omega_N(A)  \right] \right|,
\end{equation}
where we have used that $\int d\rho_{{\Lambda_{N}^{(d)}}} \,|f|^2 =\|\Psi(f)\|^2=1 $. The ground state indistinguishability result, Theorem~\ref{thm:indistinguishable}, will then be a consequence of producing an upper bound on the rate at which $\sup_{\partial \BOmega}|\mr{\omega}_N(A;\partial \BOmega)-\omega_N(A)|\to 0$ as $N\to\infty$. This will be achieved using a cluster expansion associated with a hard core polymer description of $\omega_N(A)$ and $\mr{\omega}_N(A;\partial \BOmega)$, the latter of which we now discuss.

\subsection{Graphs and weights for the hard core polymer representation} \label{sec:polymers_weights}
To bound the right hand side of \eqref{eq:bulk-groun-state-comparison}, the maps $Z_N(A)$ and $\mr{Z}_N(A;\partial \BOmega)$ will be rewritten in terms of a set of polymers and weight functions. The sets used for each map will be slightly different, and so we introduce these in a rather general setting. We begin by establishing some basic graph notation and conventions.

\begin{definition}
  Two connected subgraphs $G$ and $H$ of $\Gamma^{(d)}$ will be called \textit{(pairwise) connected}, denoted $G \disjoint H,$ if $G \cup H$ is a connected graph. Otherwise, $G$ and $H$ are \textit{not connected}, and we write $G | H$. More generally, if $\{G_i:i=1,\ldots n\}$ and $\{H_j: j=1,\ldots m\}$ are the connected components of graphs $G$ and $H$, respectively, we say $G$ and $H$ are \emph{not connected}, denoted $G|H$, if $G_i|H_j$ for all $i$ and $j$. Otherwise, $G$ and $H$ are \emph{connected} and we write $G \disjoint H$.
\end{definition}
Note that if $G$ and $H$ are connected subgraphs of $\Gamma^{(d)}$ then $G|H$ if and only if $\caV_G\cap \caV_H= \emptyset$. 
\begin{definition}\label{def:hardcore}
A collection of connected graphs $\{G_1, \ldots, G_n\}$ is said to be hard core if they are pairwise not connected, i.e. $G_k|G_l$ for all $k\neq l$.
\end{definition}


We are now ready to introduce the set of polymers of interest. The particular subgraphs of interest are connected graphs $\phi\subseteq \Gamma^{(d)}$ such that $1\leq \deg_\phi(v)\leq 2$ for all $v\in \phi$. Such graphs will be called \textit{self-avoiding}, and are partitioned into the set of \emph{closed loops} $\caC^{(d)}$, and the set of \emph{self-avoiding walks} $\caW^{(d)}$:
\begin{align}
	\caC^{(d)} &: = \{ \phi\subseteq \Gamma^{(d)} \text{ connected} : \deg_\phi(v) = 2 \; \forall \, v\in\caV_\phi\}\label{eq:all_loops} \\
	\caW^{(d)} & := \{ \phi\subseteq \Gamma^{(d)} \text{ connected} : 1\leq \deg_\phi(v) \leq 2 \; \forall \, v\in\caV_\phi\} \setminus \caC^{(d)}\label{eq:all_walks}
\end{align}
Each self-avoiding walk $\phi$ has exactly two vertices $\{v,w\}$ such that $\deg_\phi(v)=\deg_\phi(w)=1$, called the \emph{endpoints}, and all other vertices have degree two in $G$. For convenience, we will denote $\operatorname{ep}(\phi)$ the set of endpoints of a self-avoiding walk $\phi$.

The subset of self-avoiding walks of interest $\caS^{(d)}\subsetneq \caW^{(d)}$ are those whose endpoints belong to $\Gamma^{(0)}$:
\begin{equation}
  \caS^{(d)} := \{ \phi \in \caW^{(d)} : \operatorname{ep}(\phi) \subset \Gamma^{(0)} \}.
\end{equation}
The set of all possible polymers is then given by
\begin{equation}\label{eq:all_polymers}
	\caP^{(d)}:= \caC^{(d)}\cup \caS^{(d)}.
\end{equation}

Since the endpoints of every walk from $\caS^{(d)}$ belong to $\Gamma^{(0)}$, the map
\begin{equation}\label{polymer_bijection}
	\iota_d: \cP^{(d)} \to \cP^{(0)}
\end{equation}
obtained from replacing the spin-1 chain between $(v,w)\in\Gamma^{(0)}$ by an edge is a bijection. As a convention, the ``length'' of a polymer is taken to be the number of edges in its undecorated representative:

\begin{definition}\label{def:length}
	For any undecorated polymer $\phi \in \mc{P}^{(0)}$, define the \textit{length} to be the number of edges in $\phi$, i.e. $\ell(\phi) = | \caB_\phi| .$ For any decorated polymer $\phi \in\cP^{(d)}$, define the length by $\ell(\phi) := \ell(\iota_d(\phi))$ where $\iota_d$ is the bijection from \eqref{polymer_bijection}.
\end{definition}
Note that the total number of edges $|\caB_\phi|$ for any $\phi\in\caP^{(d)}$ is $(d+1)\ell(\phi)$. Since $\Gamma^{(0)}$ is bipartite, any closed loop $\phi\in\caC^{(0)}$ necessarily has even length.

We will need to consider specific subsets of $\caP^{(d)}$ in order to derive the polymer representation of $Z_N(A)$ and $\mr{Z}_N(A; \partial \BOmega)$. For the convenience of the reader, we introduce them now. 

For $0 < K < N$, let $\caC_{N,K}^{(d)}\subseteq\caC^{(d)}$ the set of closed loops in $\Lambda_N^{(d)}$ wich do not intersect $\Lambda_K^{(d)}$:
\begin{equation}
  \caC_{N,K}^{(d)} := \{ \phi \in \caC^{(d)} : \phi \subset \Lambda_N^{(d)} ,\,  \phi|\Lambda_K^{(d)} \},
\end{equation}
and $\caC_{N,0}^{(d)} := \{ \phi \in \caC^{(d)} : \phi \subset \Lambda_N^{(d)} \}$.
Moreover, let $\mc{S}_{N,K}^{(d)}\subseteq \caS^{(d)}$ be the set of self-avoiding walks with edges in $\mc{B}_N^{(d)} \setminus \mc{B}_K^{(d)}$ and endpoints in $\Lambda_K^{(d)}$:
\begin{equation}
  \mc{S}_{N,K}^{(d)} := \{ \phi \in \caS^{(d)} : \mc{B}_\phi \subset \mc{B}_N^{(d)} \setminus \mc{B}_K^{(d)} ,\, \operatorname{ep}(\phi) \subset \partial \Lambda_K^{(d)} \}.
\end{equation}
Then we define, for $K >0$,
\begin{align}\label{eq:newpolymer}
	\mc{P}_{N,K}^{(d)} &:= \caC_{N,K}^{(d)} \cup\mc{S}_{N,K}^{(d)}.
\end{align}
and $\mc{P}_{N,0}^{(d)} = \mc{C}^{(d)}_{N,0}$. This is the set of polymers we will use in the representation of $Z_N(A)$. Note that $\iota_d(\caP_{N,K}^{(d)})=\caP_{N,K}^{(0)}$.

In the representation of $\mr{Z}_N(A;\partial \BOmega)$ we will also have to consider self-avoiding walks that begin or end at $x\in\partial\Lambda_N^{(d)}$. Explicitly, for $K>0$,
we denote by
\begin{equation}
\mr{\caS}_{N,K}^{(d)} := \{ \phi \in \caS^{(d)} : \mc{B}_\phi \subset \mc{B}_N^{(d)} \setminus \mc{B}_K^{(d)} ,\, \operatorname{ep}(\phi) \subset \partial \Lambda_K^{(d)}\cup\partial\Lambda_N^{(d)}\}
\end{equation}
the set of all self-avoiding walks with edges contained in $\caB_N^{(d)}\setminus\caB_K^{(d)}$ and endpoints in $\partial\Lambda_K^{(d)}\cup\partial\Lambda_N^{(d)}$. For the special case $K=0$, we consider only the self-avoiding walks which begin and end at $\partial \Lambda_N^{(d)}$:
\begin{equation}
\mr{\caS}_{N,0}^{(d)} := \{ \phi \in \caS^{(d)} : \mc{B}_\phi \subset \mc{B}_N^{(d)} \setminus \mc{B}_K^{(d)} ,\, \operatorname{ep}(\phi) \subset \partial \Lambda_N^{(d)} \}.
\end{equation}
Then similarly to before, we define for $K \ge 0$
\begin{equation}\label{eq:polymers2}
\mr{\caP}_{N,K}^{(d)} := \caC_{N,K}^{(d)} \cup \mr{\caS}_{N,K}^{(d)}.
\end{equation}

Now that we have introduced the polymer sets of interest, we turn to introducing a weight function $W_d$ on $\cP^{(d)}$ that will express how much a give polymer contributes to the polymer representation of $Z_N(A)$ or $\mr{Z}_N(A;\partial \mathbf{\Omega})$.

The weight of any $\phi\in\caP^{(d)}$ is defined analogously to those given by Kennedy, Lieb and Tasaki in \cite{Kennedy1988}. This is a consequence of evaluating certain integrals which naturally arise when calculating ground state expectations. For any closed loop $\phi\in\caC^{(d)}$, the weight $W_d(\phi)$ is
\begin{align}\label{eq:loopweight}
	W_d(\phi):= (1/3)^{ (d+1) \ell(\phi) -1} = \int d \BOmega^\caV ~ \prod _{(i,j) \in \phi}-\Omega_i \cdot \Omega _ j .
\end{align}
where $\caV \subset \Gamma^{(d)}$ is any finite subset of vertices such that $\mc{V}_\phi \subset \caV$. Similarly, if $\phi\in\caS^{(d)}$ is a self-avoiding walk with endpoints $v,w\in\Gamma^{(0)}$, the weight function $W_d(\phi)$ is
\begin{align}
	W_d(\phi):= (-1/3)^{ (d+1)\ell(\phi)-1 }\partial \phi (\BOmega) = \int d \BOmega^{\caV} ~ \prod _{(i,j) \in \phi} -\Omega_i \cdot \Omega _ j \label{eq:walk-weight}
\end{align}
where $ \partial \phi (\BOmega) : = - \Omega_{v} \cdot \Omega_{w}$ and $\caV \subset \Gamma^{(d)}$ is any finite set of vertices such that $\mc{V}_\phi \cap \caV = \mc{V}_\phi \setminus \set{v,w}$. 

In either case above, (\ref{eq:loopweight}) and (\ref{eq:walk-weight}) and their independence of the set $\caV$ are easy to verify by first integrating $\int d \BOmega^x  = 1$ for all sites $x\in \caV\setminus \mc{V}_\phi$, and using $\deg_\phi(x)=2$ for every $x\in  \caV\cap\caV_\phi$
to evaluate
\[\int d \BOmega^{\caV\cap\caV_\phi} \prod _{(v,w) \in \phi} -\Omega_v \cdot \Omega _w = (-1)^{(d+1)\ell(\phi)}\int d \BOmega^{\caV\cap\caV_\phi}  \prod _{(v,w) \in \phi} \Omega_v \cdot \Omega _w \]
by successively applying
\begin{equation}\label{one_iteration}
	\int d \Omega_x(\Omega_y\cdot \Omega_x)(\Omega_x\cdot\Omega_z) = \frac{1}{3}\Omega_y\cdot\Omega_z .
\end{equation}
The exponents in \eqref{eq:loopweight}-\eqref{eq:walk-weight} count the number of vertices to which \eqref{one_iteration} is applied. In the case of the closed loop, integrating over the final site $x\in\mc{V}_{\phi}$ yields $\int d \Omega_x(\Omega_x\cdot\Omega_x) = 1$ as $\Omega_x$ has unit length.  The expression \eqref{one_iteration} can be explicitly computed using \eqref{integrand_data} and \eqref{eq:sphereical_coord}.

\subsection{The hard-core polymer representation of $Z_N$}\label{sec:loop}In \cite{Kennedy1988}, Kennedy, Lieb and Tasaki used a loop gas representation  of the ground state with a hard core condition to evaluate ground state expectations when $d=0$. Their methods can also be used for the decorated models. Here we prove Lemma~\ref{prop:multiexpansion2} which establishes a modified version of this representation for 
\begin{equation}\label{eq:ZN_expanded}
Z_N(A) = 2^{-|\cB_N^{(d)}|} \int d\BOmega^{\Lambda_{N}^{(d)}}~ \prod _{(i,j)\in \Lambda_{N}^{(d)}} (1-\Omega_i \cdot \Omega_j) A(\BOmega).
\end{equation}

To rewrite $Z_N(A)$, we follow \cite{Kennedy1988} and distribute the product from \eqref{eq:ZN_expanded} to find
\begin{equation}\label{eq:product_to_sum}
	\prod _{(i,j)\in \Lambda_{N}^{(d)}} (1-\Omega_i \cdot \Omega_j) = \sum_{ G\in \cG_N^{(d)}} \prod_{(i,j)\in G}(-\Omega_i \cdot \Omega_j).
      \end{equation}
where $ \cG_N^{(d)}$ is the collection of all subgraphs of $\Lambda_N^{(d)}$ with no isolated vertices
\begin{equation}
  \cG_N^{(d)} := \{ G\subseteq \Lambda_N^{(d)} : \deg_G(v)> 0 \; \forall v\in G\}\,.
\end{equation}

Inserting this into \eqref{eq:ZN_expanded}, the sum is then simplified by removing subgraphs for which
\begin{equation}\label{one_graph_int}
	\int d\BOmega^{\Lambda_{N}^{(d)}} \prod_{(i,j)\in G}(-\Omega_i \cdot \Omega_j)  A(\BOmega) = 0.
\end{equation}
The graphs that remain are characterized by their connected components, which necessarily form a hard core set.

The next result, which shows $Z_N(A)$ can be written in terms of hard core subsets of $\caP_{N,K}^{(d)}$, is a slight simplification of \cite[Equation~4.14]{Kennedy1988}. We follow their method of proof, which is a consequence of observing that the variable $\Omega$ associated with a single site satisfies
\begin{equation}\label{eq:transformation}
	\int d \Omega ~ f(-\Omega) = \int d\Omega ~ f(\Omega)\,.
\end{equation}
As the integral is taken over $[0,\pi]\times S^1$ , this can easily be verified from recognizing that the measure $d\Omega$ is invariant under the diffeomorphism $R: [0, \pi ] \times S^1 \to [0,\pi]\times S^1$ that sends $\Omega \mapsto -\Omega$ defined by $R(\theta,\phi) = (\pi - \theta, \phi + \pi)$.

\begin{figure}
	\begin{center}
		\begin{tikzpicture}[
			x=\hexsize,
			y=\hexsize,
			square/.style={draw,fill,minimum size=1.5pt,inner sep=0pt},
			dot/.style={circle,circle,fill,minimum size=2pt,inner sep=0pt}
			]
			\foreach \x in {1,...,12} {
				\coordinate(X) at (0:1.5*\x);
				\ifodd\x
				\def\ymax{7}
				\coordinate(X) at ($(X)+(0:0.5)+(-120:1)$);
				\else
				\def\ymax{6}
				\fi
				\foreach \y in {1,...,\ymax} {
					\coordinate (\x-\y) at ($(X)+(60:\y)+(120:\y)$);
					\draw[gray] (\x-\y) +(-60:1)
					\foreach \z [remember=\z as \lastz (initially 5)] in {0,...,5} {
						-- coordinate(\x-\y-\lastz-m) +(\z*60:1) coordinate(\x-\y-\z)} -- cycle;} }
			\node at (7-4) {\small{$\Lambda_2^{(1)}$}};
			\draw[line width = .75mm, draw = blue]
			(5-5-0) -- (5-5-1) -- (5-5-2) -- (5-5-3) -- (4-4-2) -- (4-4-3) -- (4-4-4)-- (4-3-3)
			-- (3-3-2) -- (3-3-3) -- (3-3-4) -- (3-3-5) -- (4-2-4) -- (4-1-3) -- (4-1-4) -- (4-1-5)
			-- (4-1-0) -- (5-1-1) -- (6-1-2) -- (5-2-1) -- (6-2-2);
			\draw[line width = .75mm, draw=green]
			(8-4-1) -- (9-5-2) -- (9-5-1) -- (9-5-0) -- (9-5-5) -- (9-4-0) -- (9-4-5)--(9-4-4);
			\draw[line width = .75mm, draw = purple]
			(7-5-1) -- (8-5-2) -- (8-5-1) -- (9-6-2) -- (9-6-1) -- (9-6-0) -- (10-5-1) --
			(10-5-0) -- (10-5-5) -- (10-4-0) -- (10-4-5) -- (10-3-0) -- (10-3-5) -- (10-3-4) 
			-- (9-3-5) -- (9-3-4) -- (9-3-3) ;  
			\draw[very thick]
			(7-5-2) node[square, red, scale=.75]{}--  (7-5-1-m) node[square, scale=.75]{}--  (7-5-1) node[square, red, scale=.75]{}--  (7-5-0-m) node[square, scale=.75]{}--  (7-5-0) node[square, scale=.75]{}--  (8-4-1-m) node[square, scale=.75]{}--  (8-4-1) node[square, red, scale=.75]{}--  (8-4-0-m) node[square, scale=.75]{}--  (8-4-0) node[square,red, scale=.75]{}--  (8-4-5-m) node[square, scale=.75]{}--  (8-4-5) node[square, scale=.75]{}--  (8-3-0-m) node[square, scale=.75]{}--  (8-3-0) node[square,red, scale=.75]{}--  (8-3-5-m) node[square, scale=.75]{}--  (8-3-5) node[square,red, scale=.75]{}--  (8-3-4-m) node[square, scale=.75]{}--  (8-3-4) node[square, scale=.75]{}--  (7-3-5-m) node[square, scale=.75]{}--  (7-3-5) node[square,red, scale=.75]{}--  (7-3-4-m) node[square, scale=.75]{}--  (7-3-4) node[square,red, scale=.75]{}--  (7-3-3-m) node[square, scale=.75]{}--  (7-3-3) node[square, scale=.75]{}--  (6-3-4-m) node[square, scale=.75]{}--  (6-3-4) node[square,red, scale=.75]{}--  (6-3-3-m) node[square, scale=.75]{}--  (6-3-3) node[square,red, scale=.75]{}--  (6-3-2-m) node[square, scale=.75]{}--  (6-3-2) node[square, scale=.75]{}--  (6-4-3-m) node[square, scale=.75]{}--  (6-4-3) node[square, red, scale=.75]{}--  (6-4-2-m) node[square, scale=.75]{}--  (6-4-2) node[square, red, scale=.75]{}--  (6-4-1-m) node[square, scale=.75]{}--  (6-4-1) node[square, scale=.75]{}--  (7-5-2-m) node[square, scale=.75]{}--  (7-5-2) node[square, red, scale=.75]{};
			\draw[very thick] 
						(7-5-0) --  (7-5-5-m) node[square, scale=.75]{}--  (7-5-5) node[square, scale=.75]{}--  (7-4-0-m) node[square, scale=.75]{}--  (7-4-0) node[square, scale=.75]{}--  (7-4-5-m) node[square, scale=.75]{}--  (7-4-5) node[square, scale=.75]{}--  (7-4-4-m) node[square, scale=.75]{}--  (7-4-4) node[square, scale=.75]{}--  (7-4-3-m) node[square, scale=.75]{}--  (7-4-3) node[square, scale=.75]{}--  (7-4-2-m) node[square, scale=.75]{}--  (7-4-2) node[square, scale=.75]{}--  (7-5-3-m) node[square, scale=.75]{}--  (7-5-3);
			\draw[very thick]
			(7-4-2) node[square, scale=.75]{}-- (7-4-1-m) node[square, scale=.75]{}-- (7-4-1) node[square, scale=.75]{};
			\draw[very thick]
			(6-4-5)-- (6-4-4-m) node[square, scale=.75]{}-- (6-4-4) node[square, scale=.75]{};
			\draw[very thick]
			(6-3-0)-- (6-3-5-m) node[square, scale=.75]{}-- (6-3-5) node[square, scale=.75]{};
			\draw[very thick]
			(7-4-5)-- (7-3-0-m) node[square, scale=.75]{}-- (7-3-0) node[square, scale=.75]{};
			\draw[very thick]
			(7-4-0)-- (8-4-4-m) node[square, scale=.75]{}-- (8-4-5) node[square, scale=.75]{};
		\end{tikzpicture}
	\end{center}
	\caption{Example of polymers from $\cS_{5,2}^{(1)}$. The decoration is suppressed outside of $\Lambda_2^{(1)}$ for clarity.}
\end{figure}
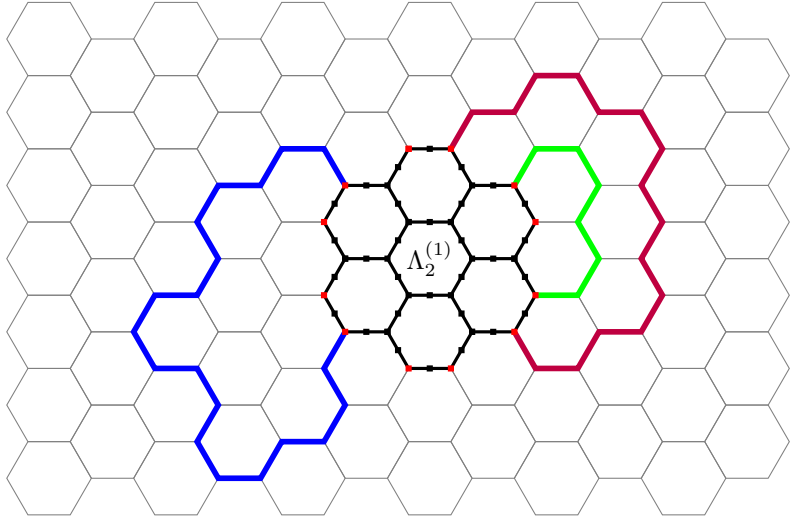

\begin{lemma}\label{prop:multiexpansion2}
	Fix $N>K\geq 0$ and $d\geq 0$. Then for any $A\in\cA_K^{(d)}$,
	\begin{equation}\label{eq:boundary-expansion}
		Z_N(A)  = \int d\rho_{\Lambda_K^{(d)}} A(\BOmega) \Phi_{N,K}(\BOmega)
	\end{equation}
	where with respect to the weight functions from \eqref{eq:loopweight}-\eqref{eq:walk-weight},
	\begin{equation}\label{eq:PhiNK}
	\Phi_{N,K}(\BOmega) :=  2^{-|\caB_N^{(d)}\setminus\caB_K^{(d)}|}\sum^{h.c.} _{ \set{\phi_1,\ldots, \phi_n}\subseteq \cP_{N,K}^{(d)} }  W_d(\phi_1)\cdots W_d(\phi_n),
	\end{equation}
and the notation indicates the sum is taken over all hard core collections $\set{\phi_1,\ldots, \phi_n} \subset \cP_{N,K}^{(d)}$.
\end{lemma}
The sum appearing in the definition of $\Phi_{N,K}$ is finite, and moreover, each product $W_d(\phi_1)\cdots W_d(\phi_n)$ is necessarily a (real-valued) monomial in the variables $\Omega_{x}$ associated with the boundary $x\in\partial \Lambda_K^{(d)}$. Also notice that $Z_N = \int d\rho_{\Lambda_{K}^{(d)}} \Phi_{N,K}(\BOmega)$.

\begin{proof} 
Fix a subgraph $G\in \caG_N^{(d)}$. If there is a $x\in G\setminus \Lambda_K^{(d)}$ such that $\deg_G(x)$ is odd, then applying \eqref{eq:transformation} to this vertex implies
\begin{equation}\label{eq:odd_degree}
	\int d\Omega_x \prod_{(x,y)\in G}(-\Omega_x \cdot \Omega_y) = 0.
\end{equation}	
As the symbol $A(\BOmega)$ is constant in the variables $\Omega_v$ for all $v\in \Lambda_N^{(d)}\setminus\Lambda_K^{(d)}$, by first integrating over this vertex $x$ one sees that $G$ satisfies \eqref{one_graph_int} by \eqref{eq:odd_degree}. By definition of $\caG_{N}^{(d)}$, the degree of any vertex $x\in G$ is at least one. Since the degree of any vertex is at most three,
\begin{equation}\label{eq:step1_hc}
Z_N(A) = 2^{-|\caB_N^{(d)}|}\sum_{G\in\caG_{N,K}^{(d)}}\int d\BOmega^{\Lambda_{N}^{(d)}} \prod_{(x,y)\in G}(-\Omega_x \cdot \Omega_y)  A(\BOmega)
\end{equation}
where $\caG_{N,K}^{(d)}=\left\{G\in\cG_N^{(d)} : \deg_G(x)=2, \; \text{for all}\; x\in G\setminus\Lambda_K^{(d)}\right\}.$

Now, fix any $G\in\caG_{N,K}^{(d)}$ and decompose $G=H\cup H'$ where $H$, respectively $H'$, is the (possibly empty) subgraph of $G$ whose edges belong to $\caB_N^{(d)}\setminus \caB_K^{(d)}$, respectively $\caB_K^{(d)}$. Note that the only vertices these graphs can have in common belong to $\partial\Lambda_K^{(d)}$.

Each connected component of $H$ is necessarily a subgraph of a connected component of $G$ and, moreover, the degree constraint on the vertices of $G$ guarantees that each connected component of $H$ is an element of $\caP_{N,K}^{(d)}$.
Now, recall that by construction, $\Lambda_K^{(d)}$ is a union of $d$-decorated hexagons, and therefore every boundary vertex $v\in\partial\Lambda_K^{(d)}$ has only one edge which does not belong to $\caB_K^{(d)}$. In other words, for each $v\in\partial\Lambda_K^{(d)}$ there exists a unique $(v,w)\in\caB_N^{(d)}\setminus\caB_K^{(d)}$. Hence, any connected component of $H$ that intersects $\partial\Lambda_K^{(d)}$ must be a self-avoiding walk. Therefore, the connected components of $H$ are a hard core subset of $\caP_{N,K}^{(d)}$. It trivially follows that $H'\in\caG_K^{(d)}$.

 Conversely, the graph union $H'\cup H$ of a graph $H'\in\caG_K^{(d)}$ and a hard core set $H\subseteq \caP_{N,K}^{(d)}$ is an element of $\caG_{N,K}^{(d)}$, and so $\caG_{N,K}^{(d)}$ is in one-to-one correspondence with such graph unions. Therefore, \eqref{eq:step1_hc} can be rewritten as
	\begin{align*}
	Z_N(A)
	& =2^{-|\cB_N^{(d)}|} \sum_{H'\in \caG_K^{(d)}}\sum^{h.c.} _{ \set{\phi_1,\ldots, \phi_n}\subseteq \cP_{N,K}^{(d)} }  \int d\BOmega^{\Lambda_N^{(d)}} ~ \prod_{i=1}^n\prod _{(v,w)\in \phi_i} (-\Omega_v\cdot \Omega_w)\prod _{(l,k) \in H'} (-\Omega_l \cdot \Omega_k) A(\BOmega).
\end{align*}

Since $A(\BOmega)$ is constant (in fact, one) on any site $x\in\Lambda_N^{(d)}\setminus\Lambda_K^{(d)}$ as $\supp(A)\subseteq \Lambda_K^{(d)}$, integrating over these sites and applying \eqref{eq:loopweight}-\eqref{eq:walk-weight} simplifies the expansion to
\begin{align}
	Z_N(A)  = 2^{-|\cB_N^{(d)}|}\sum_{H'\in \caG_K^{(d)}}  \int \, d \BOmega^{\Lambda_{K}^{(d)}}\prod_{(l,k)\in H'} (-\Omega_l \cdot \Omega_k)  A(\BOmega)\Phi_{N,K}(\BOmega),
\end{align}
from which \eqref{eq:boundary-expansion} follows from \eqref{eq:product_to_sum} and \eqref{def_rho}.
\end{proof}

\subsection{The hard-core loop representation of $\mr{Z}_{N}$}
A similar process can also be used to construct a hard core polymer representation of  $\mr Z_N(A;\partial \BOmega)$ for any $A \in \A_{K}^{(d)}$ with $K<N$, with the difference of having to consider polymers in $\mr{\caP}_{N,K}^{(d)}$ instead.

\begin{lemma}\label{lem:loop_2} Fix $N>K\geq 0$ and $d\geq 0$. Then for any $A\in\caA_{K}^{(d)}$,
	\begin{equation}\label{eq:multiexpansion3}
	\mr{Z}_N(A;\partial \BOmega) =   \int d \rho_{\Lambda_K^{(d)}} A(\BOmega) \mr{\Phi}_{N,K}(\BOmega),
	\end{equation}
where with respect to the weight functions from \eqref{eq:loopweight}-\eqref{eq:walk-weight},
\begin{equation}\label{eq:sum2}
	\mr{\Phi}_{N,K}(\BOmega) := 2^{-|\caB_N^{(d)}\setminus\caB_K^{(d)}|} \sum^{h.c.} _{ \set{\phi_1,\ldots, \phi_n}\subseteq 
		 \mr{\caP}_{N,K}^{(d)}}  W_d(\phi_1)\cdots W_d(\phi_n).
\end{equation}
\end{lemma}
The proof of this result follows analogously to that of Lemma~\ref{prop:multiexpansion2} with one small change based on noticing that $\deg_{\Lambda_N^{(d)}}(x) = 2$ for all $x\in\partial\Lambda_N^{(d)}$. Unlike the proof of the previous result, one cannot apply \eqref{eq:transformation} to these sites since the integration in \eqref{bulk_map} is only over the interior vertices. Therefore, graphs $G\in\caG_N^{(d)}$ with only one edge connecting to a boundary site $x\in\partial\Lambda_N^{(d)}$ do not necessarily satisfy \eqref{one_graph_int}. This extends the set of polymers from \eqref{eq:PhiNK} to including walks with endpoints at such sites, resulting in the definition from \eqref{eq:polymers2}.

Finally, note that $\mr{\Phi}_{N,K}$ is a function of the variables $\Omega_{v}$ for $v\in \partial \Lambda_{K}^{(d)}\cup \partial \Lambda_{N}^{(d)}$ and, similar to before, $\mr{Z}_{N}(\partial \BOmega)= \int d\rho_{\Lambda_{K}^{(d)}} \mr{\Phi}_{N,K}(\BOmega)$.

\subsection{A comparison lemma}\label{sec:comparison}
The infinite volume ground state $\omega^{(d)}$ from Theorem~\ref{thm:indistinguishable} can be obtained by showing that $(\omega_N(A))_N$ is Cauchy for each local observable $A$. This can be achieved by controlling the difference between $\Phi_{N,K}/Z_{N}$ and $\Phi_{M,K}/Z_{M}$. Lemma \ref{prop:multiexpansion2} shows that a bound depending on the essential supremum of $A(\BOmega)$ exists. To establish the LTQO condition, though, one needs a bound that depends on the operator norm of $A$, which is not immediately obvious. In general, $A$ is the compression of a multiplication operator on $L^2(X,d\BOmega)$ of a compact manifold $X$ to a finite dimensional Hilbert subspace \cite{Berezin1972}, and so one can only expect $\norm{A} \leq \norm{A(\BOmega)}_{L^\infty}$. To illustrate, let $A = \frac{1}{3} \partial_u u \in B(\HH{2})$. Using \eqref{eqn:matrix_entries}-\eqref{eq:symbol} it is easy to calculate 
\[\norm{A} = 1, \qquad \norm{A(\Omega)}_{L^\infty} = \frac{4}{3}.\] 
Given any set of degree 2 vertices $x_1,\ldots, x_n\in \Gamma^{(d)}$, the observable  $A^{(n)} = \bigotimes _{j=1}^n A_{x_j}$ still has operator norm one, but $\norm{ A^{(n)}(\BOmega)}_{L^\infty} = (4/3)^n$.

The next lemma, which is valid for operators $A \in \mr{\caA}_K^{(d)}$ supported on the interior vertices (see \eqref{notation}), can be used to recover a bound which depends on the operator norm. The key insight here is that $\mathfrak{H}_x$ can be viewed as a subspace of $L^2(d\Omega_x)$. By extending the onsite Hilbert space to $\mf{L}_x =L^2(S^2, d\Omega_x)$ at each vertex $x \in  \partial \Lambda_K^{(d)}$ and defining $\mf{L}_K^{(d)}\supset \mathfrak{H}_K^{(d)}$ by
\begin{align}\label{eq:extended_HS}
	\mf{L}_K^{(d)} = \bigotimes_{x\in \partial \Lambda_K^{(d)}} \mf{L}_x \otimes \bigotimes _{x \in \mr{\Lambda}_K^{(d)} } \mf{H}_x ,
\end{align}
one sees that $A$ acts on $\mf{L}_K^{(d)}$ by the inclusion $A \mapsto \1 \otimes A$, see the comment following \eqref{eq:IP}. Moreover, the assumption that $A$ is supported on the interior vertices guarantees that the operator norm on this space agrees with the original quantum spin model:
\[
\|\1 \otimes A\|_{\mf{L}_K^{(d)}} = \|A\|_{\mf{H}_K^{(d)}},
\]
which allows us to prove the following result.

\begin{lemma}\label{lem:comparison-lemma}
	Fix $d\geq 0$ and $0\leq K<M\leq N$. Then, for any $A\in \mr{\caA}_K^{(d)}$ ,
	\begin{equation}
		\bigg{|} \omega_N(A) - \omega_M(A) \bigg{|}
		\leq  \norm{A}~ \norm{\frac{\Phi_{N,K}}{Z_{N}} - \frac{\Phi_{M,K}}{Z_{M}}}_{L_1\left(\rho_{\Lambda_{K}}^{(d)}\right)}.
	\end{equation}
\end{lemma}

\begin{proof}
	Define two functions $\psi_1 =  \text{sign}(h) |h|^{1/2}  F$ and $\psi_2 =   |h|^{1/2} F$  where
		\begin{align*}
	h & = \frac{\Phi_{N,K}}{Z_N} - \frac{\Phi_{M,K}}{Z_M}, \qquad F = \prod _{(i,j) \in \Lambda_K^{(d)}} (u_i v_j - v _i u_j).
	\end{align*}
Recalling the change of variables (\ref{eq:change-of-coordinates}), it is easy to see that $\psi_1,\psi_2\in \mf{L}_K$ as $\Phi_{N,K}$ is a real-valued, bounded, continuous function of the angles $\theta_x,\phi_x$ defining $\Omega_x$ for $x\in \partial\Lambda_K^{(d)}$.  Note that $Z_N\neq 0$ for all $N>K$ since $\omega_N$ is a state on $\caA_K^{(d)}$ and, furthermore, $\rho_{\Lambda_K^{(d)}} = \overline{F}{F}$. As such,
	\begin{align*}
		\begin{split}
			| Z_N(A)/Z_N - Z_M(A)/Z_M| 
			& =\left|\int d \BOmega^{\Lambda_K^{(d)}}~ \rho_{\Lambda_K^{(d)}}\left[\frac{\Phi_{N,K}}{Z_N} - \frac{\Phi_{M,K}}{Z_M}\right]A(\BOmega)\right| \\
			 & = | \ip{ \psi_1, (\1 \otimes A) \psi_2}_{\mf{L}_K}|\\ 
			 &\leq \norm{A} \norm{\psi_1}_{\mf{L}_K}\norm{\psi_2}_{\mf{L}_K}.
		\end{split}
	\end{align*}
	But $\norm{\psi_1}_{\mf{L}_K} = \norm{\psi_2}_{\mf{L}_K}$ and so
	\begin{align}
		\norm{A} \norm{\psi_1}_{\mf{L}_K}\norm{\psi_2}_{\mf{L}_K} = \norm{A} \int d \BOmega^{\Lambda_K^{(d)}} \rho_{\Lambda_K^{(d)}} ~ \bigg{|} \frac{\Phi_{N,K} }{ Z_N} - \frac{\Phi_{M,K}}{Z_M} \bigg{|} . \label{eq:cauchy-schw}
	\end{align}
\end{proof}

This proof can also be adapted to compare $\omega_N(A)$ with $\braket{\Phi(f)}{A\Phi(f)}$ for any normalized ground state $\Phi(f)\in\ker(H_N^{(d)})$.

\begin{lemma}
	\label{lem:comparison-lemma-2}
	Fix $d\geq 0$ and $0\leq K<N$. Then, for any $A\in \mr{\caA}_K^{(d)}$ and normalized $\Psi(f)\in \ker H_{N}^{(d)}$,
	\begin{equation}
		\left| \omega_N(A) - \ip{\Psi(f), A \Psi(f)} \right|
		\leq  \norm{A}~ \sup_{(\Omega_x \, : \, x \in \partial \Lambda_{N}^{(d)})} \norm{\frac{\Phi_{N,K}}{Z_{N}} - \frac{\mr{\Phi}_{N,K}}{\mr Z_{N}}}_{L_1\left(\rho_{\Lambda_{K}}^{(d)}\right)}.
	\end{equation}
\end{lemma}

\begin{proof}
	Recall that $\mr{Z}_N(\partial \BOmega) \neq 0$ for any choice of the boundary variables $\partial\BOmega=(\Omega_x,x\in\partial\Lambda_N^{(d)})$, see \eqref{eq:interior_gs}. Then by \eqref{eq:bulk-groun-state-comparison},
	\[	\left| \omega_N(A) - \ip{\Psi(f), A \Psi(f)} \right|
	\leq
	 \sup_{\partial\BOmega}\left|\frac{\mr{Z}_{N}(A; \partial \BOmega)}{\mr{Z}_{N}(\partial \BOmega)} - \frac{Z_{N}(A)}{Z_{N}} \right|,
	\] 
where we use $\int d\rho_{{\Lambda_{N}^{(d)}}} \,|f|^2 = \|\Psi(f)\|^2=1$. For a fixed choice of $\partial\BOmega$, Lemma~\ref{prop:multiexpansion2}-\ref{lem:loop_2} imply that
	\[
	\left|\frac{\mr{Z}_{N}(A; \partial \BOmega)}{\mr{Z}_{N}(\partial \BOmega)} - \frac{Z_{N}(A)}{Z_{N}} \right|=\left|\int d \BOmega^{\Lambda_K^{(d)}}~ \rho_{\Lambda_K^{(d)}}\left[\frac{\mr{\Phi}_{N,K}}{\mr{Z}_N} - \frac{\Phi_{N,K}}{Z_N}\right]A(\BOmega)\right| .
	\]
	Considering again the extended Hilbert space \eqref{eq:extended_HS} and proceeding as in the proof of Lemma~\ref{lem:comparison-lemma} produces the result.
\end{proof}

We conclude this section with a final remark. If $f,g \in L_{1}(\mu)$ are such that $f>0$, $g>0$, and $\int d\mu\, g=1$, then
\begin{equation}\label{eq:l1-relative-entropy-bound}
	\norm{f-g}_{ L_{1}(\mu)} \le \norm{\frac{f}{g} - 1}_{L_\infty(\mu)} \int d\mu g =  \norm{\frac{f}{g} - 1}_{L_\infty(\mu)} \le \Ent{f}{g} e^{\Ent{f}{g}}
\end{equation}
where
\begin{equation}\label{eq:max-relative-entropy}
	\Ent{f}{g}: = \norm{\log f - \log g}_{L_{\infty}(\mu)},
\end{equation}
is the (classical) $\infty$-R\'enyi divergence\cite{Renyi1961, van_Erven_2014},
and we have used the inequality $|e^{x}-1| \le |x| e^{|x|}$. Using cluster expansion methods, we prove Theorem~\ref{thm:indistinguishable} in the next section by bounding
\begin{align}
	\Ent{\Phi_{M,K}/Z_{M}}{\Phi_{N,K}/Z_{N}} & = \sup_{(\Omega_x:x\in\partial\Lambda_K^{(d)})}\left|\log\left(\frac{\Phi_{M,K}}{Z_M}\right) -\log\left(\frac{\Phi_{N,K}}{Z_N}\right) \right| \label{ent_1}\\
	\sup_{\partial\BOmega}\Ent{\mr{\Phi}_{N,K}/\mr{Z}_{N}}{\Phi_{N,K}/Z_{N}} & = \sup_{(\Omega_x:x\in\partial\Lambda_K^{(d)}\cup\partial\Lambda_N^{(d)})}\left|\log\left(\frac{\mr{\Phi}_{N,K}}{\mr{Z}_{N}}\right) -\log\left(\frac{\Phi_{N,K}}{Z_N}\right) \right|, \label{ent_2}
\end{align}
where the supremum is taken over $\Omega_x\in S^2$ for all appropriate sites $x$.

\section{Indistinguishability of ground states}\label{sec:main}

As illustrated in \cite{Kennedy1988}, the existence of a unique infinite volume frustration-free state can be shown by transforming the hard core polymer representation into a cluster expansion. Moreover, the convergence of any finite volume ground state to this infinite volume state is exponentially fast in the distance of the observable to the boundary of the finite volume system. The goal of this section is to prove Theorem~\ref{thm:indistinguishable} which (beyond showing the existence of a unique infinite volume ground state) explicitly states the constants in these bounds and captures how the convergence depends on the support of the observable, which are necessary for applying the stability result from \cite{NSY:2022}. This will follow from appropriately bounding the cluster expansion. These bounds use a variation of a result due to Seiler \cite{Seiler1982} and require a minimum, positive decoration $d$.

\subsection{Cluster expansion preliminaries} While the sums in \eqref{eq:boundary-expansion} and \eqref{eq:sum2} are over hard core sets of polymers, the cluster expansion is a sum over \emph{clusters}, which in this work are sequences of polymers $\vec{\phi}=(\phi_1,\ldots, \phi_m)$ such that the union $G=\cup_i\phi_i$ is a connected graph. Alternatively, whether or not $\vec{\phi}$ is a cluster can be determined by its associated graph $G_{\vec{\phi}}$, whose vertex set is $\{\phi_1,\ldots, \phi_m\}$ and edge set is $\{(\phi_i,\phi_j):\phi_i \disjoint \phi_j\}$. Then $\vec{\phi}$ is a cluster if and only if $G_{\vec{\phi}}$ is connected.

For any finite set of polymers $\caP \subset \caP^{(d)}$, a convergent cluster expansion rewrites the sum over hard core collections from $\caP$ in terms of its logarithm:
\begin{equation}\label{eq:cluster-log}
  \sum _{\set{ \phi_1, \ldots, \phi_m}\subset \caP }^{h.c.} W_d(\phi_1) \cdots  W_d(\phi_m) =
  \exp\bigg(\sum _{\vec{\phi} \in \tau(\caP)} \varphi_c(\vec{\phi}) W_d^{\vec{\phi}}\bigg),
\end{equation}
where we introduce $W^{\vec{\phi}}_d := W_d(\phi_1) \cdots W_d(\phi_m)$, the set of all finite polymer sequences
\begin{align}
\tau(\caP) & = \set{ \vec{\phi} = (\phi_1, \ldots, \phi_m) \mid m\in\bN, \, \phi_i\in \caP\ \forall i },
\end{align}
and the \emph{Ursell function} defined by $\varphi_{c}(\phi_1) = 1$ and
\begin{align}
  \varphi_c( \phi_1, \ldots, \phi_m ) := \frac{1}{m!} \sum_{\substack{  G \in \mf{G}_{\vec{\phi}} :\\G \subset G_{\vec{\phi}}}} (-1)^{|\caB_G|} \quad \text{if $m\ge 2$}.
\end{align}
Above, $\mf{G}_{\vec{\phi}}$ is the set of all connected graphs on the `vertices' $\{\phi_1,\ldots, \phi_m\}$. The Ursell function is zero if $G_{\vec{\phi}}$ is not connected, and so the RHS of \eqref{eq:cluster-log} can be recognized as a sum over clusters. For a more in-depth review of cluster expansions, we point the reader to \cite{Seiler1982, Brydges1986, FV, Miracle2010}.

There are various methods for proving that the infinite series in \eqref{eq:cluster-log} converges and, hence, that this is not just a formal equality, see, e.g., \cite{ Ueltschi2004, KP1986, BZ00, BFP2010} and references within. Lemma~\ref{cor:seiler-cor} in Appendix~\ref{appendix:cluster} verifies a criterion from \cite{Ueltschi2004} which implies the convergence of the cluster expansion for any finite $\caP\subseteq \caP^{(d)}$ when $d\geq 3$. In particular, this holds for both $\caP_{N,K}^{(d)}$ and $\mr{\caP}_{N,K}^{(d)}$. With more careful counting arguments, e.g. as those in \cite{Kennedy1988}, this could be extended to $d<3$. However, this is sufficient since the application of the next result in the proof of Theorem~\ref{thm:indistinguishable} requires $d\ge 5$. This result bounds the contribution to the cluster expansion that comes from polymers whose support overlaps the support of the observable $A\in\caA_{\mr{\Lambda}_K^{(d)}}$ for which the ground state expectation is being calculated. To this end, denote by
\begin{equation}
	\Delta_K := \left\{(\phi_1,\ldots, \phi_m)\in\tau(\mr{\caP}_{N,K}^{(d)}): (\cup_i\phi_i) \disjoint \partial\Lambda_K^{(d)}\right\}.
\end{equation}
The minimum decay required on a polymer weight function to apply this result depends on
\begin{equation}\label{eq:f}
f(x) := \frac{x+1-\sqrt{x^2+1}}{x}.
\end{equation}

\begin{lemma}\label{lem:cluster-expansion-bound}
Fix $K\geq0$ and $N>\max\{K,1\}$ and let $\mu := 2\sqrt[5]{9}$. For all $k\geq 1$, $v\in\Gamma^{(0)}$ and $\phi\in \mr{\caP}_{N,K}^{(d)}$,
\begin{align}
	&\left| \{ \phi' \in  \mr{\caP}_{N,K}^{(d)} : \deg_{\phi'} (v) > 0, \, \ell(\phi')=k\} \right|  \leq \mu^k \label{eq:loop-counting1}\\
	&\left|\{\phi'\in \mr{\caP}_{N,K}^{(d)}: \phi \disjoint \phi', \, \ell(\phi')=k\}\right|  \leq \ell(\phi)\mu^k. \label{eq:loop-counting}
  \end{align}
As a consequence, for any $\epsilon>0$ and polymer weight function $w: \mr{\caP}_{N,K}^{(d)}  \to \rr$ with $|w(\phi)| \leq e^{-\beta \ell(\phi)}$ for some 
\begin{equation} \label{eq:beta_threshold}
	\beta \geq 4/e  + \ln(\mu/f(2e\mu))+\epsilon 
\end{equation}
one has
\begin{equation}\label{eq:cluster_coeff}
  \sum _{\vec{\phi}\in\Delta_K}
  \left|\varphi_c(\vec{\phi}) w^{\vec{\phi}}\right| \leq |\partial\Lambda_K^{(d)}|\frac{r(\epsilon)}{1-r(\epsilon)},
\end{equation}
where we denoted $w^{\vec{\phi}} := w(\phi_1)\cdots w(\phi_m)$ and the ratio $ r(\epsilon):= \frac{\mu f(2e\mu)e^{1-\epsilon}}{(1-\mu f(2e\mu)e^{1-\epsilon})(1-f(2e\mu)e^{-\epsilon})}<1$.
\end{lemma}

In particular, for $\epsilon = 0.03$, one has $r(\epsilon) \approx 0.9424$ and $\frac{r(\epsilon)}{1-r(\epsilon)} <17.$

\begin{proof}
	The bounds from \eqref{eq:loop-counting1}-\eqref{eq:loop-counting} are proved in Proposition~\ref{prop:mu}. As a result, if $\beta$ satisfies \eqref{eq:beta_threshold}, then Theorem~\ref{thm:seiler} holds for $\caP=\mr{\caP}_{N,K}^{(d)}$ with $\kappa = 4/e+\ln(\mu)$ which establishes \eqref{eq:cluster_coeff}.
\end{proof}

\subsection{Indistinguishability of the finite volume ground states}
\label{sec:indistinguishability-proof}
\medskip 

We can now prove that the existence of $\omega^{(d)}(A) = \lim _N Z_N(A)/Z_N$, which is necessarily a frustration-free ground state on $\A_{\Gamma^{(d)}}$ of the decorated AKLT model. The explicit bounds also provide estimates on the dependence of the convergence on the support and operator norm of $A$.

Let us start by emphasizing the role of the decoration in the cluster expansion. The geometry of the decorated lattice is such that any two connected polymers from $\caP^{(d)}$ necessarily intersect at a spin-3/2 site. Thus, a sequence $\vec{\phi}=(\phi_1,\ldots,\phi_m)$ is a cluster if and only if $\iota_d(\vec{\phi})=(\iota_d(\phi_1), \ldots \iota_d(\phi_m))$ is a cluster, where $\iota_d$ is as in \eqref{polymer_bijection}. In fact, their respective graphs $G_{\vec{\phi}}$ and $G_{\iota_d(\vec{\phi})}$ are isomorphic, and the Ursell function is invariant under $\vec{\phi}\mapsto\iota_d(\vec{\phi})$. Replacing $\ell(\phi) = \ell(\iota_d(\phi))$ in the weight functions $W_d$ from \eqref{eq:loopweight}-\eqref{eq:walk-weight} also allows us to consider $W_d$ as acting on $\caP^{(0)}$. Thus, it is only the decay rate of the weight function $W_d$ that depends nontrivially of $d$ in the cluster expansion \eqref{eq:cluster-log}. This is the crucial observation for all of the remaining proofs. As such, we will henceforth suppress the dependence on $d$ for any sets involving polymers. To simplify notation, we set
\begin{align}\label{eq:finite_sequences}
\tau_{N,K} &:= \tau(\caP^{(d)}_{N,K}), & \mr{\tau}_{N,K} &:= \tau(\mr{\caP}^{(d)}_{N,K}).
\end{align}
\begin{lemma}\label{lem:partition-difference}
  Fix $d\geq 5$ and $M\geq N > K\geq 1$, and set $\partial\BOmega=(\Omega_x:x\in\partial\Lambda_N^{(d)})$.
  Then \begin{equation}\label{eq:entropy_bound}
	\max\left\{\Ent{\frac{\Phi_{M,K}}{Z_M}}{\frac{\Phi_{N,K}}{Z_N}}, \; \sup_{\partial\BOmega}\Ent{\frac{\mr{\Phi}_{N,K}}{Z_N}}{\frac{\Phi_{N,K}}{Z_N}}\right\}
	\le 17|\partial\Lambda_K^{(d)}|e^{-2\alpha(d)(N-K)}
  \end{equation}
	where $\alpha(d) :=d\ln(3)-4/e-\ln(\mu/f(2e\mu))-.03$ with $\mu=2\sqrt[5]{9}$ and $f(x)$ as in \eqref{eq:f}.
      \end{lemma}
  The assumption that $d\geq 5$ guarantees $\alpha(d)>0$. In particular, $\alpha(5)\approx 0.0032$.
\begin{proof}
We first bound $\Ent{\tfrac{\Phi_{M,K}}{Z_M}}{\tfrac{\Phi_{N,K}}{Z_N}} $. For all $N>K\geq 0$, the cluster expansion for 
\[
\Phi_{N,K} := 2^{-|\caB_N^{(d)}\setminus\caB_K^{(d)}|}\sum _{\set{ \phi_1, \ldots, \phi_m}\subset \caP_{N,K} }^{h.c.} W_d(\phi_1) \cdots  W_d(\phi_m)
\]
converges where $\caB_0^{(d)}=\emptyset$ by Lemma~\ref{cor:seiler-cor}. In particular, this implies $\Phi_{N,K}>0$, and so applying \eqref{eq:cluster-log} to both $\Phi_{N,K}$ and $Z_N=\Phi_{N,0}$ yields
	\begin{align}\label{normalized_partition}
	\log\left(\frac{\Phi_{N,K}}{Z_N}\right) & = \log(2^{|\caB_K^{(d)}|})+ \sum _{\vec{\phi} \in \tau_{N,K}} \varphi_c(\vec{\phi})W_d^{\vec{\phi}} - \sum _{\vec{\phi} \in \tau_{N,0}} \varphi_c(\vec{\phi})W_d^{\vec{\phi}},
	\end{align}
where the sequences of polymers \eqref{eq:finite_sequences} are defined with respect to \eqref{eq:newpolymer}. Since $\tau_{N,K} \subset \tau_{M,K}$, \eqref{ent_1} can be simplified to
	\begin{align}
	  \Ent{\frac{\Phi_{M,K}}{Z_M}}{\frac{\Phi_{N,K}}{Z_N}} 
	  & = \sup_{(\Omega_x : x\in\partial\Lambda_K^{(d)})}\left|  \sum _{\vec{\phi} \in \tau_{M,K}  \setminus \tau_{N,K} } \varphi_c(\vec{\phi}) W_d^{\vec{\phi}} - \sum _{\vec{\phi} \in \tau_{M,0}  \setminus \tau_{N,0}}  \varphi_c(\vec{\phi}) W_d^{\vec{\phi}}\right| . \label{eq:cancel_1}
	\end{align}

To further cancel common terms, partition $\tau_{M,K} \setminus \tau_{N,K}$ into two sets depending on if the sequence $\vec{\phi}=(\phi_1,\ldots,\phi_m)$ contains a walk, and partition $\tau_{M,0} \setminus \tau_{N,0}$ by whether or not $\vec{\phi}$ has a loop intersecting $\Lambda_K^{(d)}$. This produces
\[
\tau_{M,K} \setminus \tau_{N,K} = (\mc{L}_{M,K} \setminus \mc{L}_{N,K}) \cup \Delta^1 _{M,N}, \qquad 
\tau_{M,0} \setminus \tau_{N,0} = (\mc{L}_{M,K} \setminus \mc{L}_{N,K})  \cup \Delta^2_{M,N}
\]
where the sets are taken as
	\begin{align*}
	\mc{L}_{N,K} & = \set{ \vec{\gamma} = (\gamma_1, \ldots, \gamma_m) : \forall i, \gamma_i \in \mc{C}_N , \gamma_i  | \Lambda_K^{(d)}} \\
	\Delta^1_{M,N} & = \set{ \vec{\phi}\in \tau_{M,K} \setminus \tau_{N,K}  : \exists j, \phi_j \in \mc{S}_{M,K}\setminus \mc{S}_{N,K}} \\
	\Delta^2_{M,N} & = \set{ \vec{\phi} \in \tau_{M,0} \setminus\tau_{N,0} : \exists j, \phi_j \cap \Lambda_K^{(d)} \not = \emptyset}.
	\end{align*}
Thus, the difference in \eqref{eq:cancel_1} further reduces, and one finds
	\begin{align}\label{eq:cancel2}
	\Ent{\frac{\Phi_{M,K}}{Z_M}}{\frac{\Phi_{N,K}}{Z_N}}  \leq \sup_{(\Omega_x : x\in\partial\Lambda_K^{(d)})} \Bigg(\sum _{\vec{\phi}\in \Delta^1_{N,M}\cup\Delta^2_{N,M}} |\varphi_c (\vec{\phi}) W_d^{\vec{\phi}}|\Bigg) .
	\end{align}

Applying Lemma~\ref{lem:cluster-expansion-bound} on can further bound \eqref{eq:cancel2} uniformly in the boundary variables. To see this, notice that any cluster $\vec{\phi}=(\phi_1,\ldots,\phi_m)\in\Delta_{M,N}^{1}\cup\Delta_{M,N}^{2}$ produces a connected graph $G=\cup_j \phi_j$ that intersects both $\partial\Lambda_N^{(d)}$ and $\partial\Lambda_K^{(d)}$. Hence, 
	\[\sum_{j=1}^m\ell(\phi_j) \geq D_0(\partial\Lambda_N^{(d)}, \partial\Lambda_K^{(d)})=2(N-K)\] 
	where $D_0$ is the graph distance on $\Gamma^{(0)}$. Recall the definition of $W_d$ from \eqref{eq:loopweight}-\eqref{eq:walk-weight}, and note that $|\partial S(\BOmega)| = |\Omega_a\cdot\Omega_b| \leq 1$ for any walk $S$ with endpoints $a,b\in\partial\Lambda_K^{(d)}$. As such,
	\begin{align}
	|W_d^{\vec{\phi}}| & \leq e^{ - 2\alpha(N-K)} \prod_{i=1}^m w_\alpha (\phi_i), \qquad w_\alpha(\phi) := e^{ -(d\ln(3)-\alpha) \ell(\phi)}.
	\end{align}
	Since $\varphi_c$ is only nonzero on clusters and $\Delta_{M,N}^{1}\cup\Delta_{M,N}^{2} \subseteq \Delta_K$, applying Lemma~\ref{lem:cluster-expansion-bound} with $\beta = d\ln(3)-\alpha$ and $\epsilon = 0.03$ shows that for any choice of the boundary variables $(\Omega_x : x\in\partial\Lambda_K^{(d)})$,
	\begin{align}\label{eq:new_weight}
\sum _{\vec{\phi}\in \Delta^1_{N,M}\cup\Delta^2_{N,M}} |\varphi_c (\vec{\phi}) W_d^{\vec{\phi}}| & \leq e^{ - 2\alpha(N-K)}  \sum _{\vec{\phi}\in \Delta_K} |\varphi_c (\vec{\phi}) w_\alpha^{\vec{\phi}}| \leq 17|\Lambda_K^{(d)}|e^{ - 2\alpha(N-K)}.
	\end{align}
Here, we have also used that $\frac{r(0.03)}{1-r(0.03)}<17.$ The desired bound then follows by \eqref{eq:cancel2}.
	
To bound \eqref{ent_2}, fix a choice of the boundary variables and use the cluster expansion \eqref{eq:cluster-log} to write
\[	
\log\left(\frac{\mr{\Phi}_{N,K}}{\mr{Z}_N}\right) = \log(2^{|\caB_K^{(d)}|}) +
\sum _{\vec{\phi} \in \mr{\tau}_{N,K} } \varphi_c(\vec{\phi}) W_d^{\vec{\phi}} - \sum _{\vec{\phi} \in \mr{\tau}_{N,0}} \varphi_c(\vec{\phi}) W_d^{\vec{\phi}},
\]
where we recall \eqref{eq:polymers2}-\eqref{eq:sum2}. Let $\partial\Lambda_{N,K}=\partial\Lambda_N^{(d)}\cup\partial\Lambda_K^{(d)}$. Noticing that $\tau_{N,K}\subseteq \mr{\tau}_{N,K}$, an analogous procedure to the previous case shows 
\begin{align}
\sup_{\partial\BOmega}\Ent{\frac{\mr{\Phi}_{N,K}}{Z_N}}{\frac{\Phi_{N,K}}{Z_N}}
= &
\sup_{(\Omega_x : x\in\partial\Lambda_{N,K})}~ \left|\sum _{\vec{\phi} \in \tau_{N,K}^\circ\setminus\tau_{N,K} } \varphi_c(\vec{\phi}) W_d^{\vec{\phi}} - \sum _{\vec{\gamma} \in \tau_{N,0}^\circ\setminus\tau_{N,0}  }  \varphi_c(\vec{\gamma}) W_d^{\vec{\gamma}} \right| \nonumber \\
\leq & 
\sup_{(\Omega_x : x\in\partial\Lambda_{N,K})}~e^{-2\alpha(N-K)}\sum _{\vec{\phi} \in \mr{\Delta}_N^1\cup\mr{\Delta}_N^2} |\varphi_c(\vec{\phi}) w_\alpha^{\vec{\phi}} |
\end{align}
where the weight function is as in \eqref{eq:new_weight} and the final summation set is taken over the union of
	\begin{align*}
	\mr{\Delta}_N^1 : = & \left\{\vec{\phi}\in \tau_{N,K}^\circ\setminus\tau_{N,K} : \exists j, \phi_j \in \mr{\caS}_{N,K}\text{ has endpoints } v\in\partial\Lambda_N^{(d)}, \, w\in\partial\Lambda_K^{(d)} \right\} \\
	\mr{\Delta}_N^2 : = & \left\{\vec{\phi}\in \tau_{N,0}^\circ\setminus\tau_{N,0} : \exists i,j, \;\phi_i \disjoint \Lambda_K^{(d)}, \, \phi_j \in \mr{\caS}_{N,0}\right\} .
\end{align*}
Since $\mr{\Delta}_N^1\cup\mr{\Delta}_N^2\subseteq \Delta_K$, the result for this case again follows from Lemma~\eqref{lem:cluster-expansion-bound}. Combining the two bounds produces \eqref{eq:entropy_bound}.
\end{proof}

Using Lemma~\ref{lem:partition-difference}, we are now able to prove Theorem~\ref{thm:indistinguishable}, the main result regarding local indistinguishability of ground states of the decorated AKLT model.

  \begin{proof}[Proof of Theorem~\ref{thm:indistinguishable}]
	We begin by proving the existence of a frustration-free ground state $\omega^d:\caA_{\Gamma^{(d)}}\to \bC$. Let $A$ be a local observable and fix $K>0$ so that $A\in\mr{\caA}_K^{(d)}$. The result follows from showing $(\omega_N(A))_{N>K}$ is Cauchy. By Lemma~\ref{lem:comparison-lemma}, 
	\[|\omega_M(A)-\omega_N(A)| \leq \norm{A}~ \norm{\frac{\Phi_{N,K}}{Z_{N}} - \frac{\Phi_{M,K}}{Z_{M}}}_{L_1\left(\rho_{\Lambda_{K}}^{(d)}\right)}\]
	for any $M\geq N >K$. Again, since the weight function $W_d$ is real-valued, the cluster expansion \eqref{eq:cluster-log} shows that $\Phi_{N,K}>0$ and $Z_N = \Phi_{N,0}=\int d\rho_{\Lambda_K^{(d)}}\Phi_{N,K}>0$. Thus, \eqref{eq:l1-relative-entropy-bound} holds and by Lemma \ref{lem:partition-difference}
	\begin{equation}\label{eq:limit_bound}
	|\omega_M(A)-\omega_N(A)| \leq F_{\alpha}(N,K)e^{F_{\alpha}(N,K)} 
	\end{equation}
where we recall that $F_{\alpha}(N,K)=17|\partial \Lambda_K|e^{-\alpha(N-K)}$. Hence, $\omega_N(A)$ is Cauchy and the limit $\omega^{(d)}(A):=\lim_N\omega_N(A)$ exists. Moreover, this is necessarily a frustration-free ground state by the discussion following \eqref{eq:unnormalized}.

Now, fix $A\in\mr{\caA}_K^{(d)}$ and let $\Psi_N \in \ker(H_N^{(d)})$ be arbitrary with $N>K$. Then by \eqref{eq:limit_bound}
\begin{align*}
|\omega^{(d)}(A)-\ip{\Psi_N, A\Psi_N}| 
	& \leq |\omega^{(d)}(A) - \omega_{N}(A)| + |\omega_N(A)-\ip{\Psi_N, A\Psi_N}| \\ 
	& \le F_{\alpha}(N,K) e^{F_{\alpha}(N,K)} \norm{A} + |\omega_N(A)-\ip{\Psi_N, A\Psi_N}|.
\end{align*}
To bound the second term, the cluster expansion once again guarantees $\mr{\Phi}_{N,K}>0$ and $\mr{Z}_N>0$. Thus, applying Lemma~\ref{lem:comparison-lemma-2} and \eqref{eq:l1-relative-entropy-bound} produces
\[|\omega_N(A)-\ip{\Psi_N, A\Psi_N}|\leq  \norm{A}~ \sup_{(\Omega_x \, : \, x \in \partial \Lambda_{N}^{(d)})} \norm{\frac{\Phi_{N,K}}{Z_{N}} - \frac{\mr{\Phi}_{N,K}}{\mr Z_{N}}}_{L_1\left(\rho_{\Lambda_{K}}^{(d)}\right)},\]
and the result follows from a second application of \eqref{eq:l1-relative-entropy-bound}  and Lemma~\ref{lem:partition-difference}.
  \end{proof}

\medskip 

\noindent \textbf{Acknowledgments.} A.M. was supported by the EU Horizon 2020 programme under the Marie Sklodowska-Curie grant agreement No. 101023822.  A.L. was supported by grant RYC2019-026475-I funded by MCIN/AEI/\-10.13039/501100011033 and by “ESF Investing in your future”, by grants PID2020-113523GB-I00 and CEX2019-000904-S funded by MCIN/AEI/\-10.13039/501100011033, and by Comunidad de Madrid (grant QUITEMAD-CM, ref. P2018/TCS-4342).
A.Y. was supported by the DFG under EXC-2111--390814868. A.M. thanks Bergfinnur Durhuus helpful discussions, and all of the authors thank Bruno Nachtergaele for helpful discussions and comments on this work.

\section{Appendix}\label{sec:appendix}

 \subsection{Modifications to the stability argument}\label{appendix:hexagon_lattice} We outline the necessary modifications for adapting the stability argument from \cite{NSY:2021}, which is stated for perturbations supported on balls of the lattice, $b_x(n)$, to the case considered in this work, which considers perturbations supported on the volumes $\Lambda_n^{(d)}(\tilde{x})$ from \eqref{eq:Lambda_n}. The overall goal is to provide the interested reader who wishes to go through \cite{NSY:2021} with sufficient information to see that, indeed, the result of that work applies to this slightly different context.
 
 The modifications to the Assumptions for \cite[Theorem 2.8]{NSY:2021} were discussed in Section~\ref{sec:stability}. It was also discussed there that the main reason this result still holds in present context is that the volumes $\Lambda_n^{(d)}(\tilde{x})$ satisfy similar properties to the balls $b_x(n)$ of the lattice considered in \cite{NSY:2021}. The main criterion that is left to verify is that the Heisenberg dynamics generated by the perturbations considered in this work satisfies a Lieb-Robinson bound governed by a function that decays like a stretched exponential. This is the key result that ensures all of the quasi-local estimates that are used in the stability argument are still valid. Criteria that imply Lieb-Robinson bounds are well-known and prolific in the literature. Proposition~\ref{prop:LR_bound} shows that one such criteria is satisfied for the class of perturbations from this work, see \cite[Theorem 2.3]{NSY:2022}. Before proving this result, though, we first outline the other small cosmetic changes one needs to adapt the proof.

 The remaining changes correspond to appropriately replacing certain set throughout in \cite{NSY:2021} with the corresponding set from this work. Specifically, one should replace the ball of radius $n$ with $\Lambda_n^{(d)}(\tilde{x})$ in \cite[Equation 2.27]{NSY:2021} and in all other places that the ground state projections are considered. Moreover, the ball of radius $n$ should be replaced by $\mr{\Lambda}_n^{(d)}(\tilde{x})$ in the case of the localizing operators \cite[Equation 3.2]{NSY:2021} and all subsequent objects these are used to define. One should also use the dual lattice $\tilde{\Gamma}^{(0)}$ whenever considering the two sets that index the perturbation. 
 
 We now turn to proving that the perturbation satisfies a sufficient Lieb-Robinson bound, which necessarily must be given with respect to the distance in $\Gamma^{(d)}$ rather than the distance in the dual lattice $\tilde{\Gamma}^{(0)}$. As shown in \cite{NSY:2022}, such a Lieb-Robinson bounds guarantees that the quasi-local estimates from \cite[Equations (4.10)-(4.11)]{NSY:2021} hold. This is also sufficient to verify \eqref{eq:pert_generator}. For concreteness, we consider the graph distance on $\Gamma^{(d)}$, denoted by $D_d$, in the result below.
\begin{proposition}\label{prop:LR_bound}
Suppose that $\Phi(\tilde{x},n)^* = \Phi(\tilde{x},n) \in \caA_{\Lambda_n^{(d)}(\tilde{x})}$ is such that for all $\tilde{x}$ and $n$ 
\begin{equation}\label{int_decay}
	\|\Phi(\tilde{x},n)\|\leq \|\Phi\|e^{-an^{\theta}}
\end{equation}
for some constants $a,\|\Phi\|>0$ and $0 \leq \theta < 1$. Then, for any $a'<a$ and $p\geq 0$
\begin{equation}\label{eq:F_norm}
\sup_{x,y\in\Gamma^{(d)}}\sum_{\substack{\tilde{z},n: \\ x,y\in\Lambda_n^{(d)}(\tilde{z})}}|\Lambda_n^{(d)}(\tilde{z})|\|\Phi(\tilde{z},n)\| \leq C(d,a-a',p)\|\Phi\|\frac{e^{-a'\left(\frac{D_d(x,y)}{4(d+1)}+\frac{1}{4}\right)^\theta}}{\left(\frac{D_d(x,y)}{4(d+1)}+\frac{1}{4}\right)^p}.
\end{equation}
where
\begin{equation}\label{eq:constant}
C(d,a-a',p) = 81(3d+2)\sum_{n\geq 1}n^{p+4}e^{-(a-a')n^\theta}.
\end{equation} 
\end{proposition}

The proof of Proposition~\ref{prop:LR_bound} uses some simple counting estimates on finite subsets of $\Gamma^{(d)}$, which are proved in Proposition~\ref{lem:boundary-length} below, as well as the following relation between the graph distance $D_d$ on $\Gamma^{(d)}$ and the graph distance $\tilde{D}$ on the dual lattice $\tilde{\Gamma}^{(0)}$. Namely,
\begin{equation}\label{eq:distance_relation}
\frac{D_d(x,y)}{2(d+1)}-\frac{3}{2} \leq \tilde{D}(\tilde{x},\tilde{y}) \leq \frac{D_d(x,y)}{2(d+1)}
\end{equation}
where for each $x\in\Gamma^{(d)}$, $\tilde{x}\in\tilde{\Gamma}^{(0)}$ is any dual lattice site such that $x\in \Lambda_{1}^{(d)}(\tilde{x})$, i.e. $x$ belongs to the plaquette associated with $\tilde{x}$. 
\begin{proof}
	Let $\caI_x=\{\tilde{x}\in\tilde{\Gamma}^{(0)}: x\in \Lambda_1^{(d)}(\tilde{x})\}$ for each $x\in\Gamma^{(d)}$, and notice that
	\[
	x\in \Lambda_n^{(d)}(\tilde{z}) \iff \tilde{D}(\tilde{x},\tilde{z}) \leq n-1 \;\; \text{for some} \;\; \tilde{x}\in\caI_x.
	\]
	 Then, for any fixed $x,y\in\Gamma^{(d)}$,
	\begin{align}
		\sum_{\substack{\tilde{z},n: \\ x,y\in\Lambda_n^{(d)}(\tilde{z})}}|\Lambda_n^{(d)}(\tilde{z})|\|\Phi(\tilde{z},n)\| 
		& \leq
		\sum_{\substack{\tilde{x}\in\caI_x \\ \tilde{y}\in\caI_y}}
		\sum_{\substack{\tilde{z},n: \\ \tilde{x},\tilde{y}\in b_{n-1}(\tilde{z})}}|\Lambda_n^{(d)}(\tilde{z})|\|\Phi(\tilde{z},n)\| \nonumber\\
		& \leq 	
		\sum_{\substack{\tilde{x}\in\caI_x \\ \tilde{y}\in\caI_y}}
		\sum_{\substack{\tilde{z},n: \\ \tilde{z}\in b_{n-1}(\tilde{x})\cap b_{n-1}(\tilde{y})}}|\Lambda_n^{(d)}(\tilde{z})|\|\Phi(\tilde{z},n)\| \label{lrbound1}
	\end{align}
where $b_n(\tilde{z})=\{\tilde{y}\in\tilde{\Gamma}^{(0)}:\tilde{D}(\tilde{z},\tilde{y})\leq n\}$. The intersection constraint requires $n-1\geq \frac{\tilde{D}(\tilde{x},\tilde{y})}{2}$, and so the sum is further bounded by
\begin{align}
\sum_{\substack{\tilde{z},n: \\ z\in b_{n-1}(\tilde{x})\cap b_{n-1}(\tilde{y})}}|\Lambda_n^{(d)}(\tilde{z})|\|\Phi(\tilde{z},n)\| 
&\leq 
3(3d+2)\|\Phi\|\sum_{n\geq  \frac{\tilde{D}(\tilde{x},\tilde{y})}{2} +1} \sum_{\tilde{z}\in b_{n-1}(\tilde{x})}  n^2e^{-an^{\theta}}\label{lrbound2}
\end{align}
where one uses \eqref{int_decay} and Proposition~\ref{lem:boundary-length}. It is easy to show that $|b_{n-1}(\tilde{x})| =1+6\sum_{k=1}^{n-1}k \leq 3n^2$, which combined with \eqref{lrbound1}-\eqref{lrbound2} produces
\[
\sum_{\substack{\tilde{x}\in\caI_x \\ \tilde{y}\in\caI_y}}	\sum_{\substack{\tilde{z},n: \\ \tilde{x},\tilde{y}\in b_{n-1}(\tilde{z})}}|\Lambda_n^{(d)}(\tilde{z})|\|\Phi(\tilde{z},n)\| \leq
\left( 9(3d+2)\|\Phi\|\sum_{n\geq 1}n^{p+4}e^{-(a-a')n^\theta}\right) \sum_{\substack{\tilde{x}\in\caI_x \\ \tilde{y}\in\caI_y}}	\frac{e^{-a'(\tilde{D}(\tilde{x},\tilde{y})/2+1)^\theta}}{(\tilde{D}(\tilde{x},\tilde{y})/2+1)^p}.
\]
Inserting \eqref{eq:distance_relation} and using $|\caI_x|\leq 3$ produces the final result.

\end{proof}

To conclude, we prove some simple volume estimates on the decorated hexagonal lattice that are used in the previous result and throughout this work.

\begin{proposition} \label{lem:boundary-length}
	Let $|\Lambda|$ denote the number of sites in $\Lambda$. Then, for all $n\geq 1$ and $d\geq 0$
	\begin{equation}
	 |\partial \Lambda_n^{(d)}| = 6n \quad\text{and}\quad |\Lambda_n^{(d)}|  = 3(3d+2)n^2-3dn.
	\end{equation}
\end{proposition}

\begin{proof} All boundary sites of $\Lambda_n^{(d)}$ lie on plaquettes $\Lambda_1^{(d)}(\tilde{x})$ such that $\tilde{D}(\tilde{x},\tilde{0}) = n-1$. Dividing $\Lambda_n^{(d)}$ into six pieces, it is easy to count that there are precisely $6(n-1)$ such plaquettes, of which $6$ contribute two boundary sites while the remaining $6(n-2)$ contribute one.
	
For each boundary site $v\in\partial\Lambda_{n-1}^{(d)}$, let $w_v\in \Lambda_{n}^{(d)}\setminus\Lambda_{n-1}^{(d)}$ denote the spin-3/2 site so that $(w_v,v)$ is an edge in $\Lambda_{n}^{(0)}$, and define $Y_{w_v}$ to be volume consisting of $w_v$ and the $3d$ spin-1 sites emanating from $w_v$. Then there exists a set $S_n$ of $6d$ spin-1 vertices so that
	\[
	\Lambda_n^{(d)}\setminus\Lambda_{n-1}^{(d)} = \biguplus_{w_v: v\in\partial\Lambda_{n-1}^{(d)}} Y_{w_v} \uplus \partial\Lambda_n^{(d)}   \uplus S_n.
	\]
Thus, for all $n\geq 2$ the previous case implies 
\[|\Lambda_n^{(d)}\setminus\Lambda_{n-1}^{(d)} |= 6(n-1)(3d+1) + 6n + 6d = 6(3d+2)n-6(2d+1).\]
The result follows from summing $|\Lambda_n^{(d)}|=|\Lambda_1^{(d)}|+\sum_{k=2}^{n}|\Lambda_k^{(d)}\setminus\Lambda_{k-1}^{(d)} |.$
\end{proof}

\subsection{Cluster expansion preliminaries}\label{appendix:cluster}

The goal of this section is to establish some basic cluster expansion results. In the standard cluster expansion theory, quantities are defined in terms of a function $g:\caP^{(d)}\times \caP^{(d)} \to \{0,-1\}$ such that $g(\phi,\phi)=-1$. In this work, we have taken
\begin{equation} \label{eq:activity}
g(\phi,\phi') = 0 \;\;\text{if}\;\; \phi | \phi', \quad  g(\phi,\phi') = -1 \;\;\text{if}\;\; \phi \disjoint \phi'
\end{equation}
which can be used to rewrite both the hard-core and cluster conditions as well as the Ursell function:
\[	\varphi_c( \phi_1, \ldots, \phi_m ) =  \frac{1}{m!} \sum _{ G \in \mf{G}_{\vec{\phi}}} \prod _{(\phi, \phi') \in G} g(\phi, \phi'), \qquad m\geq 2.\]

\medskip 

 It is necessary to control the number of polymers of a given length $k$ that intersect a fixed polymer $\phi$ to prove convergence of a cluster expansion. The following such bound was improved using a careful counting argument by Kennedy, Lieb and Tasaki in the appendix of \cite{Kennedy1988} for the case $d=0$. However, the stability result relies on Theorem~\ref{thm:seiler} which requires $d\ge 5$, and so these more straightforward bounds suffice.

\begin{lemma}\label{lem:number-polymers}
	Fix $d\geq 0$ and let $g$ be as in \eqref{eq:activity}. For any spin-3/2 vertex $v\in\Gamma^{(0)}$, the number of polymers of length $k\geq 1$ containing $v$ is bounded by
	\begin{align}
	\left| \{ \phi \in \mc{P}^{(d)} : \ell(\phi)=k, \deg_\phi (v) > 0\} \right| & \leq 3(k+1) 2^{k-2} . \label{eq:loopwalkcount}
	\end{align}
	Consequently, for any $\phi\in\mc{P}^{(d)}$ 
	\begin{equation}\label{eq:intersecting_bound}
	\left|\{\phi'\in\mc{P}^{(d)}:  \ell(\phi')=k, \, g(\phi,\phi')=-1\}\right| \leq 3\left(\ell(\phi)+1\right)(k+1)2^{k-2}.
	\end{equation}
\end{lemma}

\begin{proof}
	Given Definition~\ref{def:length} and \eqref{polymer_bijection}, it is sufficient to prove \eqref{eq:loopwalkcount} for $d=0$.

	There are two cases: if $v$ is the endpoint of $\phi$ (in which case, $\phi$ is a self-avoiding walk), or if $v$ is not. In the first case, there is exactly one edge from $\phi$ emanating from $v$. There are three choices for this edge. Each of the remaining $k-1$ edges can be chosen by successively laying an edge at the endpoint of the previously laid edge. There are two choices for the new edge, and so there are at most $3\cdot 2^{k-1}$ self-avoiding walks with $v$ as an endpoint.

	If $v$ is not the endpoint of $\phi$, then there are two edges $(w_1,v),\, (w_2,v)\in \phi$. For $i=1,2$, denote by $k_i\geq 1$ the number of edges in $\phi$ that connect to $v$ through $w_i$, and note that if $\phi$ is a loop one can assume WLOG that $k_1=1$. There are three choices for $\{w_1,\,w_2\}$ after which there are at most $2^{k_i-1}$ walks that connect to $v$ through $w_i$. Since $k_1+k_2=k$, summing over the possible values of $k_1$ shows there are at most $3(k-1)2^{k-2}$	polymers with two edges connecting at $v$. Combining this with the previous case produces \eqref{eq:loopwalkcount}.

	The result for \eqref{eq:intersecting_bound} follows from realizing that two connected polymers $\phi,\phi'\in\caP^{(d)}$ must intersect at a spin-3/2 site. There are at most $\ell(\phi)+1$ such sites in $\phi$.
\end{proof}

The bound from \eqref{eq:intersecting_bound} can be used to prove convergence of the cluster expansion for $d\geq 3$ when the weight functions are taken as in \eqref{eq:loopweight}-\eqref{eq:walk-weight}. Note that the product involving $g$ in the following result encodes the hard core polymer condition, see \eqref{eq:activity} and Definition~\ref{def:hardcore}.

\begin{lemma}\label{cor:seiler-cor}
	Fix $d\geq 3$ and suppose $\caP\subseteq\caP^{(d)}$ and $w: \caP \to \rr$ is a weight function such that $|w(\phi)| \leq (1/3)^{(d+1)\ell(\phi)-1}$. Then, for $g$ as in \eqref{eq:activity}
	\begin{align}\label{eq:cluster_covergence}
		\sum _{\set{\phi_1,\ldots,\phi_n }\subseteq \caP} w(\phi_1)\cdots w(\phi_n) \prod _{1 \leq i < j \leq n} (1+g(\phi_i,\phi_j)) = \exp \bigg{(} \sum _{\vec{\phi}\in \tau(\caP)} \varphi_c(\vec{\phi}) w^{\vec{\phi}} \bigg{)}. 
	\end{align}
\end{lemma}

\begin{proof}
	It is well-known that the above cluster expansion will converge, i.e. \eqref{eq:cluster_covergence} holds, if there is an $\ep>0$ such that 
	\begin{align}\label{eq:criterion}
		\sup_\phi \frac{1}{\ell(\phi)} \sum _{\substack{\sigma \in \caP \\ g(\phi, \sigma)=-1 }} |w(\sigma)| e^{ \ep \ell(\sigma)} < \ep ,
	\end{align}
	see, e.g., \cite[Theorem 1]{Ueltschi2004}. Inserting the assumed weight bound, breaking up the sum in terms of the length $\ell(\sigma)=k$ and applying \eqref{eq:intersecting_bound} shows this criterion holds if
	\begin{align}
		\sup_{\phi}\frac{9(\ell(\phi)+1)}{4\ell(\phi)}\sum_{k=1}^{\infty}(k+1)\left(\frac{2e^{\ep}}{3^{d+1}}\right)^k < \ep.
	\end{align}
	As $\ell(\phi)\geq 1$, a simple geometric series argument implies this is satisfied if  
	\[
	\frac{9}{9+2\ep} < \left(1-\frac{2e^\ep}{3^{d+1}}\right)^2
	\]
	for some $0 <\ep < (d+1)\ln(3)-\ln(2)$. This inequality holds, e.g., when $\ep=1$ and $d\geq 3$.
\end{proof}

\medskip

Theorem~\ref{thm:seiler} below adapts a result due to Seiler, see \cite[Theorem 3.13]{Seiler1982}, to produce a bound on the logarithm of the cluster expansion which makes explicit the dependence on the support of the observable. This proof makes use of the following bound due to Malyshev. An elementary proof of this bound can also be found in \cite[Theorem 3.2]{Seiler1982}.

\begin{theorem}[Malyshev \cite{Malyshev1978}]\label{thm:malyshev} Let $\caP\subseteq \caP^{(d)}$ be a finite set of polymers and suppose $\mu>1$ is such that for all $\phi\in\caP$,
	\begin{equation}\label{eq:mu_bound}
	\left|\{\phi'\in\caP: g(\phi,\phi')=-1, \, \ell(\phi')=\ell\}\right| \leq \ell(\phi)\mu^{\ell}.
	\end{equation}
Let $X_{\vec{\phi}}(\phi)\in\bN_0$ be the number of times the polymer $\phi$ appears in $\vec{\phi}\in\tau(\caP)$. There exists a positive $\kappa \leq 4/e + \ln(\mu) $ so that
	\begin{align}
	| \varphi_c(\vec{\phi})| \leq \frac{1}{m!} \prod_{\phi\in\caP}X_{\vec{\phi}}(\phi)!e^{ \kappa \ell(\phi)  X_{\vec{\phi}} (\phi)}  \qquad \forall\,  \vec{\phi}=(\phi_1,\ldots,\phi_m)\in\tau(\caP).
	\end{align}
\end{theorem}

The function $X_{\vec{\phi}}$ from Theorem~\ref{thm:malyshev} can be used to identify the exact subset of polymers that appear in the sequence $\vec{\phi}$. Conversely, any nonzero function $X:\caP \to \bN_0$ generates a sequence $\vec{\phi}^X\in\tau(\caP)$ where each polymer $\phi$ appears $X(\phi)$ number of times. (This can be defined unambiguously by fixing an ordering of $\caP$.) This correspondence is not one-to-one, but rather satisfies
\[
\tau_X:=\{\vec{\phi}\in\tau(\caP): X_{\vec{\phi}}=X\} = \{\vec{\phi}^{X}_\pi : \pi \in S_{m(X)}\}
\]
where $m(X) = \sum_{\phi}X(\phi)$ is the length of $\vec{\phi}^X$, and $\vec{\phi}_\pi := (\phi_{\pi(1)},\ldots,\phi_{\pi(m)})$ for any permutation $\pi\in S_m$. As the Ursell function is invariant under permutations, the cluster expansion for any polymer weight function $w:\caP \to \bR$ satisfies
\begin{equation}\label{eq:cluster_equivalence}
\sum _{\vec{\phi}\in \tau(\caP)} |\varphi_c(\vec{\phi}) w^{\vec{\phi}}| 
= \sum_{X \neq 0} \sum_{\vec{\phi}\in\tau_X}|\varphi_c(\vec{\phi})w^{\vec{\phi}}|
= \sum_{X \neq 0}~\frac{m!}{\prod_{\phi}X(\phi)!}|\varphi_c(\vec{\phi}^{X})w^{\vec{\phi}^{X}}|
\end{equation}
where $w^{\vec{\phi}}:= w(\phi_1)\ldots w(\phi_m)$ for any $\vec{\phi}=(\phi_1,\ldots,\phi_m)\in\tau(\caP)$. This alternate form is used to prove the following result. Given for any finite set $\caP\subset \caP^{(d)}$, this bounds the contribution to the cluster expansion that comes from clusters $\vec{\phi}\in\tau(\caP)$ that intersect the boundary of $ \partial\Lambda_K^{(d)}$. As such, set
\begin{equation}\label{Delta_K}
	\Delta_{K} = \{\vec{\phi}\in\tau(\caP): \cup_i\phi_i \disjoint \partial\Lambda_K^{(d)} \}.
\end{equation}
The proof also depends on a function
\[
f(x) = \frac{x+1-\sqrt{x^2+1}}{x}
\]
which is used to ensure convergence of a certain geometric series.

\begin{theorem}\label{thm:seiler}
	Let $\caP\subseteq \caP^{(d)}$ be finite, and suppose $\mu>1$ satisfies the condition of Theorem~\ref{thm:malyshev} and
	\begin{equation}\label{eq:site_bound}
	|\{\phi\in\caP : \ell(\phi)=k, \, v\in\phi\}| \leq \mu^k, \quad \forall v\in\Gamma^{(0)}, \, k \geq 1.
	\end{equation}
If $w:\caP\to\bR$ is a polymer weight function such that $w(\phi) \leq e^{ - \beta \ell(\phi)}$ for some 
	\[\beta\geq \beta_\epsilon:=\kappa-\ln(f(2e\mu))+\epsilon\] 
with $\epsilon>0$, then for the set $\Delta_K$ defined as in \eqref{Delta_K} one has
	\begin{align}\label{eq:log-sum}
	\sum _{ \vec{\phi}\in\Delta_{K} } \left| \varphi_c(\vec{\phi}) w^{\vec{\phi}}\right| \leq|\partial\Lambda_K^{(d)}|\frac{r_\mu(\epsilon)}{1-r_\mu(\epsilon)}
	\end{align}
where  $r_\mu(\epsilon) :=\frac{\mu f(2e\mu)e^{1-\epsilon}}{(1-\mu f(2e\mu)e^{1-\epsilon})(1-f(2e\mu)e^{-\epsilon})}$.
\end{theorem}
As will be evident from the proof, the quantity above necessarily satisfies $r_\mu(\epsilon)<1$.

\begin{proof}
	Applying \eqref{eq:cluster_equivalence}, it follows immediately from the definition of $\Delta_K$ that
	\begin{align}
	 \sum _{ \vec{\phi}\in\Delta_{K} } \left|\varphi_c(\vec{\phi}) w^{\vec{\phi}}\right|  \leq \sum _{X \neq 0}^{\partial\Lambda_K^{(d)}} \frac{m!}{\prod_{\phi}X(\phi)!}|\varphi_c(\vec{\phi}^{X})w^{\vec{\phi}^{X}}|
	\end{align}
	where the summation notation indicates that there exists $1\leq i\leq m(X)$ so that $	\phi_i^X \disjoint \partial\Lambda_K^{(d)}$. Said differently, $\phi_i^X$ contains a site from $\partial\Lambda_K^{(d)}$.
	
	Recall that $X(\phi)\in\bN_0$ is the number of times the polymer $\phi$ appears in the sequence $\vec{\phi}^X$, and the Ursell function is only nonzero on clusters. Applying Theorem~\ref{thm:malyshev} and using the assumed bound on the weights, 
	\begin{align}
	\sum _{X \neq 0}^{\partial\Lambda_K^{(d)}} \frac{m!}{\prod_{\phi\in\caP}X(\phi)!}\left|\varphi_c(\vec{\phi}^{X})w^{\vec{\phi}^{X}}\right|  
	& \leq \sum _{\substack{X \neq 0 : \\ \vec{\phi}^X \text{cluster}}}^{\partial\Lambda_K^{(d)}} \prod _{\phi\in\caP} e^{ (\kappa - \beta) \ell( \phi) X(\phi) } \nonumber\\
	& \leq \sum _{m=1}^\infty \frac{1}{m!} \sum _{\substack{\mathrm{clusters}\\(\phi_1 ,\ldots, \phi_m)\in\tau(\caP)}}^{\partial\Lambda_K^{(d)}}   \sum _{ \substack{n_1, \ldots, n_m \geq 1}  } \prod _{j = 1}^m e^{ -(\beta - \kappa) n_j\ell(\phi_j)}
	\end{align}
	where, similar to above, the summation notation indicates that $\phi_i \disjoint \partial\Lambda_K^{(d)}$ for some $i$. Summing the geometric series over $n_j$ and observing that $\ell(\phi_j)\geq 1$, one arrives at 
	\begin{align}
	\begin{split}
	\sum _{X \neq 0}^{\partial\Lambda_K^{(d)}} \frac{m!}{\prod_{\phi\in\caP}X(\phi)!}\left|\varphi_c(\vec{\phi}^{X})w^{\vec{\phi}^{X}}\right|  
	& \leq \sum _{m=1}^\infty  \frac{(1-e^{ - (\beta - \kappa)}) ^{-m} }{m!} \sum_{k_1, \ldots, k_m\geq 1}  \sum _{\substack{\mathrm{clusters}\\(\phi_1 ,\ldots, \phi_m): \\ \ell(\phi_j)=k_j}}^{\partial\Lambda_K^{(d)}} \prod _{j=1}^m e^{ - (\beta - \kappa) k_j} .
	\end{split}
	\end{align}
	
	To approximate the sum over clusters, let $1\leq i\leq m$ be such that $\phi_i\disjoint \partial\Lambda_K^{(d)}$. By \eqref{eq:site_bound} there are at most $|\partial\Lambda_K^{(d)}|\mu^{k_i}$ possible $\phi_i$. Since $G_{\vec{\phi}}$ is connected for any cluster $\vec{\phi}=(\phi_1,\ldots,\phi_m)$, it contains a labeled tree on $m$-vertices as a subgraph. Let $\caT_m$ be the collection of all such trees, and set $d_j(T)$ to be the degree of $\phi_j$ in $T\in\caT_m$. Then, interpreting any $T\in\caT_m$ as being rooted at $\phi_i$ and applying \eqref{eq:mu_bound} with $\ell(\phi_j)=k_j$ to determine the possible number of descendants of each vertex, one finds that there are at most
	\[
	\mu^{\sum_{j\neq i}k_j}k_i^{d_i(T)}\prod_{j\neq i}k_j^{d_j(T)-1}
	\]
	clusters $\vec{\phi}$ such that $T\subseteq G_{\vec{\phi}}$. By Cayley's Theorem, there are $m^{m-2}$ labeled trees on $m$ vertices. Moreover, for any degree sequence $(d_1,\ldots,d_m)\in\bN^m$, i.e. $\sum_j (d_j-1) = m-2$, there are $\binom{m-2}{d_1-1,\ldots,d_m-1}$ labeled trees $T\in\caT_m$ with $d_j(T)=d_j$. As such,
	\begin{align} 
	\sum _{\substack{\mathrm{clusters}\\(\phi_1 ,\ldots, \phi_m): \\ \ell(\phi_j)=k_j}}^{\partial\Lambda_K^{(d)}} 1 &
	\leq |\partial\Lambda_K^{(d)}|\mu^{\sum_ik_i}\sum_{i=1}^mk_i\sum_{\substack{(d_1, \ldots d_m)\in\bN^m \\ \sum_jd_j = 2(m-1)}}\binom{m-2}{d_1-1,\ldots, d_m-1} \prod_{j=1}^mk_j^{d_j-1}\nonumber\\
	& =
	 |\partial\Lambda_K^{(d)}|\mu^{ \sum _{j=1}^m k_j}\left(\sum_{j=1}^m k_j\right)^{m-1} \nonumber \\
	 &<  |\partial\Lambda_K^{(d)}| m! \prod_{j=1}^me^{(\ln(\mu)+1)k_j}.\label{eq:sum-seiler}
	\end{align}
	Thus, from (\ref{eq:sum-seiler}) it follows:
	\begin{align*}
	\sum _{m=1}^\infty  \frac{(1-e^{ \kappa - \beta}) ^{-m} }{m!} \sum_{k_1, \ldots, k_m\geq 1}  \sum _{\substack{\mathrm{clusters}\\(\phi_1 ,\ldots, \phi_m): \\ \ell(\phi_j)=k_j}}^{\partial\Lambda_K^{(d)}}  \prod _{j=1}^m e^{ (\kappa-\beta) k_j} 
	& \leq  |\partial\Lambda_K^{(d)}|\sum_{m=1}^{\infty} \sum_{k_1, \ldots, k_m\geq 1}  \prod_{j=1}^m \frac{e^{(\kappa - \beta+\ln(\mu)+1)k_j}}{(1-e^{ \kappa - \beta})^{m}}\\
	& = |\partial\Lambda_K^{(d)}|\sum _{m \geq 1} \left( \frac{ e^{ \kappa-\beta + \ln(\mu)+1}}{\left(1-e^{\kappa-\beta + \ln(\mu)+1}\right)\left(1-e^{\kappa-\beta}\right)} \right)^m.
	\end{align*}
The above geometric series converges if and only if $\beta > \kappa-\ln(f(2e\mu))$, and the ratio is a decreasing function of $\beta$. The proof is complete from inserting $\beta_\epsilon$ to produce an explicit upper bound for the series in terms of
\[
 r_\mu(\epsilon):=\frac{ e^{ \kappa-\beta_\epsilon + \ln(\mu)+1}}{\left(1-e^{\kappa-\beta_\epsilon + \ln(\mu)+1}\right)\left(1-e^{\kappa-\beta_\epsilon}\right)}.
 \] 
\end{proof}

While the rough bounds from Lemma~\ref{lem:number-polymers} are sufficient for convergence of the cluster expansion, tighter bounds for $\ell(\phi)\leq 4$ are necessary to apply Theorem~\ref{thm:seiler} for proving the indistinguishability result for $d\geq 5$. This is the content of our final result. Here, we restrict our attention to the largest set of polymers of interest for the stability result: $\mr{\mathcal{P}}_{N,K}^{(d)}$.

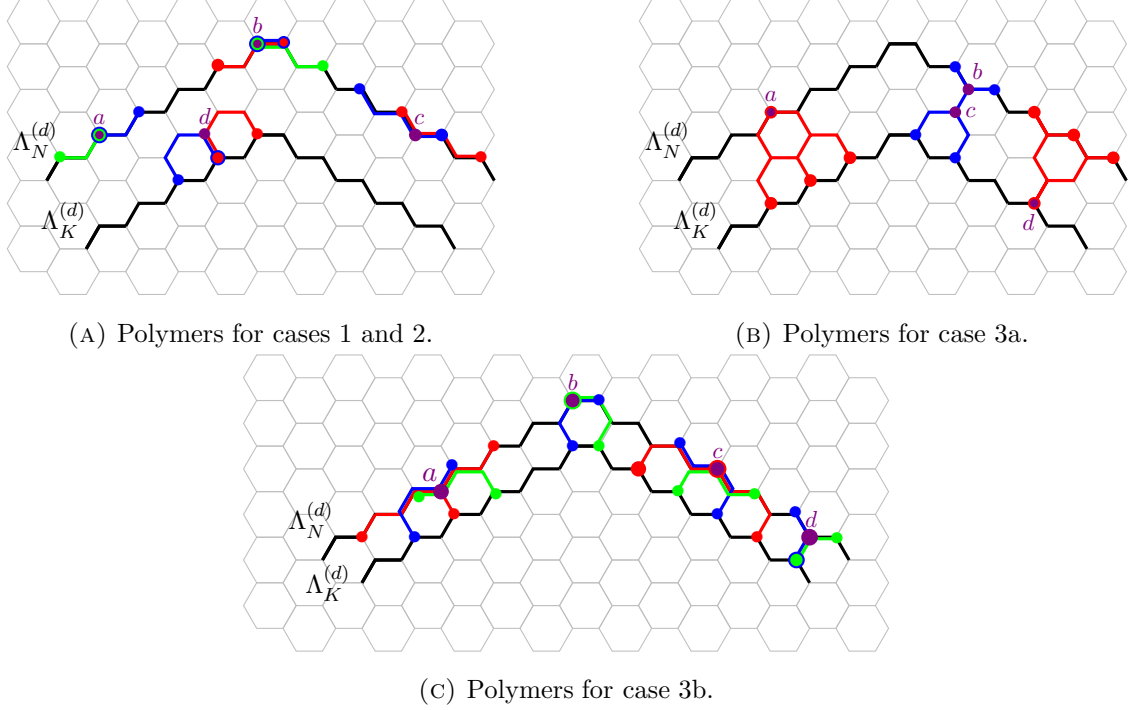
\begin{figure}
	\begin{subfigure}[b]{0.47\textwidth}
		\begin{center}
			\begin{tikzpicture}
				\def\n{.35};
				\newcommand{\hexcoord}[2]
				{[shift=(60:#1),shift=(120:#1),shift=(0:#2),shift=(-60:#2),shift=(0:#2),shift=(60:#2)]}
				\foreach \x in {1,...,6}
				\foreach \y in {2,...,7}{
					\draw[color=black!25]\hexcoord{\x*\n}{\y*\n}
					(0:\n)--(60:\n)--(120:\n)--(180:\n)--(-120:\n)--(-60:\n)--cycle;
					\draw[color=black!25,shift=(-60:\n), shift=(0:\n)]\hexcoord{\x*\n}{\y*\n}
					(0:\n)--(60:\n)--(120:\n)--(180:\n)--(-120:\n)--(-60:\n)--cycle;
				}			
				\foreach \x in {2.5,3,3.5,4,4.5,5}
				{\draw[very thick, color=black]\hexcoord{\x*\n}{\x*\n}
					(60:\n)--(120:\n)--(180:\n);
					\draw[very thick, color=black]\hexcoord{7.5*\n-\x*\n}{2.5*\n+\x*\n}
					(120:\n)--(60:\n)--(0:\n);}
				\draw[very thick, color=black]\hexcoord{2.5*\n}{2.5*\n}	(60:\n)--(120:\n)--(180:\n) node[above,xshift=-1.5mm, yshift=1.5mm]{\scalebox{.9}{$\Lambda_N^{(d)}$}};
				
				\foreach \x in {3,3.5,4,4.5,5}
				{\draw[very thick, color=black, shift=(240:2*\n), shift=(300:2*\n)]\hexcoord{\x*\n}{\x*\n}
					(60:\n)--(120:\n)--(180:\n);
					\draw[very thick, color=black, shift=(180:3*\n), shift=(300:\n), shift=(240:\n)]\hexcoord{7*\n-\x*\n}{3*\n+\x*\n}
					(120:\n)--(60:\n)--(0:\n);}
				\draw[very thick, color=black,  shift=(240:2*\n), shift=(300:2*\n)]\hexcoord{3*\n}{3*\n}	(60:\n)--(120:\n)--(180:\n) node[above,xshift=-3mm]{\scalebox{.9}{$\Lambda_K^{(d)}$}};
				
				\draw[very thick, color=blue]\hexcoord{3.5*\n}{3.5*\n} (180:2*\n) -- (180:\n) -- (120:\n)  node[blue]{$\bullet$} ;
				\draw[very thick, color=green]\hexcoord{3*\n}{3*\n} (180:2*\n) node[green]{$\bullet$} -- (180:\n) -- (120:\n);
				\draw[thick, fill=blue, draw=blue, shift=(120:\n)]\hexcoord{3*\n}{3*\n}circle (2.75pt);
				\draw[thick, fill=violet, draw=green, shift=(120:\n)]\hexcoord{3*\n}{3*\n}circle (2pt) node[violet, above]{\scalebox{.8}{$a$}};
				
				\draw[very thick, blue,yshift=.45mm]\hexcoord{5*\n}{5*\n} (60:\n) -- (120:\n);
				\draw[very thick, red]\hexcoord{5*\n}{5*\n} (60:\n) -- (120:\n)--(180:\n)--(180:2*\n);
				\draw[very thick, green, yshift=-.45mm]\hexcoord{5*\n}{5*\n} (120:\n)-- (60:\n);
				\draw[very thick, green]\hexcoord{5*\n}{5*\n} (60:\n) --(0:\n)--(0:2*\n) node[green]{$\bullet$};
				\draw[thick, fill=blue, draw=blue, shift=(120:\n)]\hexcoord{5*\n}{5*\n}circle (2.75pt);
				\draw[thick, fill=violet, draw=green, shift=(120:\n)]\hexcoord{5*\n}{5*\n}circle (2pt) node[violet, above]{\scalebox{.8}{$b$}};
				\draw[thick, fill=red, draw=blue, shift=(60:\n), yshift=.2mm]\hexcoord{5*\n}{5*\n}circle (2pt);
				\draw[thick, fill=red, draw=red, shift=(120:\n), yshift=.2mm]\hexcoord{4.5*\n}{4.5*\n}circle (2pt);
				
				\draw[very thick, color=blue, xshift=-.2mm, yshift=-.2mm]\hexcoord{3.5*\n}{6.5*\n} (120:2*\n) node[blue, xshift=-.2mm, yshift=.2mm]{$\bullet$} -- (120:\n) -- (60:\n)  -- (0:\n) -- (0:2*\n);
				\draw[very thick, color=red, yshift=.2mm,xshift=.2mm]\hexcoord{3*\n}{7*\n} (120:2*\n) node[red,  yshift=-.2mm,xshift=-.2mm]{$\bullet$}-- (120:\n) -- (60:\n)  -- (0:\n) -- (0:2*\n)node[red,  yshift=-.2mm,xshift=-.2mm]{$\bullet$};
				\draw[thick, fill=blue, draw=blue, shift=(60:\n)]\hexcoord{3*\n}{7*\n}circle (2pt);
				\draw[thick, fill=violet, draw=violet, shift=(120:\n)]\hexcoord{3*\n}{7*\n}circle (2pt) node[above, violet, xshift=.5mm]{\scalebox{.8}{$c$}};
				
				\draw[very thick, color=blue, shift=(240:\n), shift=(300:\n),xshift=-.18mm, yshift=-.18mm]\hexcoord{4*\n}{4*\n} (0:\n) -- (60:\n);
				\draw[very thick, color=blue, shift=(240:\n), shift=(300:\n)]\hexcoord{4*\n}{4*\n} (60:\n) -- (120:\n)  -- (180:\n) -- (240:\n) node[blue]{$\bullet$} ;
				\draw[very thick, color=red, shift=(240:\n), shift=(300:\n),xshift=.18mm, yshift=.18mm]\hexcoord{4.5*\n}{4.5*\n}  (180:\n) 	-- (240:\n);
				\draw[very thick, color=red, shift=(240:\n), shift=(300:\n)]\hexcoord{4.5*\n}{4.5*\n} (0:\n) node[red]{$\bullet$} -- (60:\n) -- (120:\n)  -- (180:\n)  node[violet]{$\bullet$} node[above, violet]{\scalebox{.8}{$d$}}  ;
				\draw[thick, fill=red, draw=blue, shift=(240:2*\n), shift=(300:\n)]\hexcoord{4.5*\n}{4.5*\n} circle (2.5pt) ;

			\end{tikzpicture}
			
		\end{center}
		\caption{Polymers for cases 1 and 2.}
		\label{fig:case1and2_polymers}
	\end{subfigure}\hspace{.5cm}
	\begin{subfigure}[b]{0.47\textwidth}
			\begin{center}
			\begin{tikzpicture}
				\def\n{.35};
				\newcommand{\hexcoord}[2]
				{[shift=(60:#1),shift=(120:#1),shift=(0:#2),shift=(-60:#2),shift=(0:#2),shift=(60:#2)]}
				\foreach \x in {1,...,6}
				\foreach \y in {2,...,7}{
					\draw[color=black!25]\hexcoord{\x*\n}{\y*\n}
					(0:\n)--(60:\n)--(120:\n)--(180:\n)--(-120:\n)--(-60:\n)--cycle;
					\draw[color=black!25,shift=(-60:\n), shift=(0:\n)]\hexcoord{\x*\n}{\y*\n}
					(0:\n)--(60:\n)--(120:\n)--(180:\n)--(-120:\n)--(-60:\n)--cycle;
				}			
				\foreach \x in {2.5,3,3.5,4,4.5,5}
				{\draw[very thick, color=black]\hexcoord{\x*\n}{\x*\n}
					(60:\n)--(120:\n)--(180:\n);
					\draw[very thick, color=black]\hexcoord{7.5*\n-\x*\n}{2.5*\n+\x*\n}
					(120:\n)--(60:\n)--(0:\n);}
				\draw[very thick, color=black]\hexcoord{2.5*\n}{2.5*\n}	(60:\n)--(120:\n)--(180:\n) node[above,xshift=-1.5mm, yshift=1.5mm]{\scalebox{.9}{$\Lambda_N^{(d)}$}};
				
				\foreach \x in {3,3.5,4,4.5,5}
				{\draw[very thick, color=black, shift=(240:2*\n), shift=(300:2*\n)]\hexcoord{\x*\n}{\x*\n}
					(60:\n)--(120:\n)--(180:\n);
					\draw[very thick, color=black, shift=(180:3*\n), shift=(300:\n), shift=(240:\n)]\hexcoord{7*\n-\x*\n}{3*\n+\x*\n}
					(120:\n)--(60:\n)--(0:\n);}
				\draw[very thick, color=black,  shift=(240:2*\n), shift=(300:2*\n)]\hexcoord{3*\n}{3*\n}	(60:\n)--(120:\n)--(180:\n) node[above,xshift=-3mm]{\scalebox{.9}{$\Lambda_K^{(d)}$}};
				
				\draw[very thick, color=red]\hexcoord{3.5*\n}{3.5*\n} (0:\n) -- (60:\n) -- (120:\n)  -- (180:\n) -- (240:\n) ;
				\draw[very thick, color=red, shift=(240:\n), shift=(300:\n)]\hexcoord{3.5*\n}{3.5*\n} (0:\n) -- (60:\n) -- (120:\n)  -- (180:\n) -- (240:\n) ;
				\draw[very thick, color=red, shift=(240:\n), shift=(300:\n)]\hexcoord{4*\n}{4*\n} (0:\n) -- (60:\n) -- (120:\n)  -- (180:\n);
				\draw[thick, fill=violet, draw=red, shift=(120:\n)]\hexcoord{3.5*\n}{3.5*\n}circle (2pt) node[violet, above]{\scalebox{.8}{$a$}};
				\draw[thick, fill=red, draw=red, shift=(240:\n), shift=(300:\n), shift=(240:\n)]\hexcoord{3.5*\n}{3.5*\n} circle (2pt);
				\draw[thick, fill=red, draw=red, shift=(240:\n), shift=(300:\n), shift=(0:\n)]\hexcoord{3.5*\n}{3.5*\n} circle (2pt);
				\draw[thick, fill=red, draw=red, shift=(240:\n), shift=(300:\n), shift=(0:\n)]\hexcoord{4*\n}{4*\n} circle (2pt);
				
				\draw[very thick, color=blue]\hexcoord{4*\n}{6*\n} (120:2*\n) node[blue,yshift=-.1mm] {$\bullet$} -- (120:\n)-- (60:\n) node[blue,yshift=-.1mm] {$\bullet$}-- (120:\n) -- (180:\n) -- (240:\n) -- (240:2*\n) node[blue,yshift=-.1mm] {$\bullet$} ;
				\draw[very thick, color=blue]\hexcoord{4.5*\n}{5.5*\n} (-60:\n) -- (-120:\n)-- (-120:2*\n) node[blue,yshift=-.1mm] {$\bullet$} ;
				\draw[fill=violet, draw=violet, shift=(0:\n)]\hexcoord{4.5*\n}{5.5*\n} circle (2pt) node[violet, above,xshift=1.2mm]{\scalebox{.8}{$b$}};
				\draw[fill=violet, draw=violet, shift=(-60:\n)]\hexcoord{4.5*\n}{5.5*\n} circle (2pt) node[violet, right]{\scalebox{.8}{$c$}};
				
				\draw[very thick, color=red]\hexcoord{3*\n}{7*\n} (0:\n) -- (60:\n) -- (120:\n)  -- (120:2*\n) -- (120:\n) -- (180:\n) -- (240:\n) -- (240:2*\n) -- (240:\n) -- (300:\n)-- (0:\n) -- (0:2*\n) ;
				\draw[thick, fill=red, draw=red, shift=(60:\n)]\hexcoord{3*\n}{7*\n} circle (2pt);
				\draw[thick, fill=red, draw=red, shift=(0:2*\n)]\hexcoord{3*\n}{7*\n} circle (2pt);
				\draw[thick, fill=red, draw=red, shift=(120:2*\n)]\hexcoord{3*\n}{7*\n} circle (2pt);
				\draw[thick, fill=violet, draw=red, shift=(2400:2*\n)]\hexcoord{3*\n}{7*\n} circle (2pt) node[violet, below,xshift=-.75mm]{\scalebox{.8}{$d$}};
			\end{tikzpicture}
			
		\end{center}
		\caption{Polymers for case 3a.}
		\label{fig:case3a_polymers}
	\end{subfigure}
\begin{subfigure}{\textwidth}
		\begin{center}
		\begin{tikzpicture}
			\def\n{.35};
			\newcommand{\hexcoord}[2]
			{[shift=(60:#1),shift=(120:#1),shift=(0:#2),shift=(-60:#2),shift=(0:#2),shift=(60:#2)]}
			\foreach \x in {1,...,6}
			\foreach \y in {1,...,8}{
				\draw[color=black!25]\hexcoord{\x*\n}{\y*\n}
				(0:\n)--(60:\n)--(120:\n)--(180:\n)--(-120:\n)--(-60:\n)--cycle;
				\draw[color=black!25,shift=(-60:\n), shift=(0:\n)]\hexcoord{\x*\n}{\y*\n}
				(0:\n)--(60:\n)--(120:\n)--(180:\n)--(-120:\n)--(-60:\n)--cycle;
			}			
			\foreach \x in {2,2.5,3,3.5,4,4.5,5}
			{\draw[very thick, color=black]\hexcoord{\x*\n}{\x*\n}
				(60:\n)--(120:\n)--(180:\n);
				\draw[very thick, color=black]\hexcoord{7*\n-\x*\n}{3*\n+\x*\n}
				(120:\n)--(60:\n)--(0:\n);}
			\draw[very thick, color=black]\hexcoord{2*\n}{2*\n}(60:\n)--(120:\n)--(180:\n) node[above,xshift=-1.5mm, yshift=1.5mm]{\scalebox{.9}{$\Lambda_N^{(d)}$}};
			
			\foreach \x in {2.5,3,3.5,4,4.5,5}
			{\draw[very thick, color=black, shift=(240:\n), shift=(300:\n)]\hexcoord{\x*\n}{\x*\n}
				(60:\n)--(120:\n)--(180:\n);
				\draw[very thick, color=black, shift=(180:3*\n)]\hexcoord{6.5*\n-\x*\n}{3.5*\n+\x*\n}
				(120:\n)--(60:\n)--(0:\n);}
			\draw[very thick, color=black,  shift=(240:\n), shift=(300:\n)]\hexcoord{2.5*\n}{2.5*\n}	(60:\n)--(120:\n)--(180:\n) node[left]{\scalebox{.9}{$\Lambda_K^{(d)}$}};
			
			\draw[very thick, color=red]\hexcoord{3*\n}{3*\n} (0:\n) node[red]{$\bullet$} -- (60:\n) -- (120:\n)  -- (180:\n) -- (180:2*\n);
			\draw[very thick, color=red]\hexcoord{2.5*\n}{2.5*\n} (120:\n)--(180:\n) node[red]{$\bullet$}; 
			\draw[very thick, color=red]\hexcoord{3.5*\n}{3.5*\n} (180:\n)-- (120:\n) -- (60:\n) -- (60:2*\n) node[red]{$\bullet$};
			\draw[very thick, color=blue, yshift=.4mm, xshift=-.3mm]\hexcoord{3*\n}{3*\n} (240:\n) node[blue, xshift=.3mm, yshift=-.4mm]{$\bullet$} -- (180:\n) -- (120:\n)  -- (60:\n) -- (60:2*\n) node[blue]{$\bullet$};
			\draw[very thick, color=green, yshift=-.4mm, xshift=.3mm]\hexcoord{3.5*\n}{3.5*\n} (0:\n) node[green]{$\bullet$} -- (60:\n) -- (120:\n)  -- (180:\n) -- (180:2*\n) node[green, xshift=.3mm, yshift=-.4mm]{$\bullet$};
			\draw[fill=violet, draw=violet, shift=(180:\n)]\hexcoord{3.5*\n}{3.5*\n}  circle (2.75pt) node[above, xshift=-1.5mm, violet] {$a$};
			
			\draw[very thick, color=green,yshift=.4mm,xshift=-.15mm]\hexcoord{5*\n}{5*\n}(300:\n) node[green, yshift=-.4mm,xshift=.15mm]{$\bullet$}-- (0:\n) --  (60:\n)  -- (120:\n);
			\draw[very thick, color=blue]\hexcoord{5*\n}{5*\n} (240:\n) node[blue]{$\bullet$} -- (180:\n) -- (120:\n)  -- (60:\n) node[blue]{$\bullet$};
			\draw[thick, draw=green, fill=violet, shift=(120:\n)]\hexcoord{5*\n}{5*\n} circle (3pt) node[above,violet]{\scalebox{.8}{$b$}};
			
			\draw[very thick, color=red]\hexcoord{4*\n}{6*\n} (180:\n)-- (120:\n)--(60:\n)--(0:\n)--(0:2*\n);
			\draw[very thick, color=red]\hexcoord{3*\n}{7*\n} (120:2*\n)-- (120:\n)--(60:\n)--(0:\n)--(300:\n)node[red]{$\bullet$};
			\draw[very thick, color=blue,xshift=.3mm,yshift=.4mm]\hexcoord{3.5*\n}{6.5*\n} (120:2*\n) node[blue]{$\bullet$}--(120:\n)--(60:\n)--(0:\n)--(300:\n) node[blue,xshift=-.3mm,yshift=-.4mm]{$\bullet$};
			\draw[very thick, color=green, yshift=-.4mm,xshift=-.3mm]\hexcoord{3.5*\n}{6.5*\n} (180:\n) node [green, yshift=.4mm,xshift=.3mm]{$\bullet$}--(120:\n)--(60:\n);
			\draw[very thick, color=green, yshift=-.4mm,xshift=-.3mm]\hexcoord{3*\n}{7*\n} (120:2*\n)--(120:\n)--(60:\n) node [green]{$\bullet$};
			\draw[thick, draw=red, fill=red, shift=(300:\n)]\hexcoord{4.5*\n}{5.5*\n} circle (2.5pt);
			\draw[thick, draw=red, fill=violet, shift=(60:\n)]\hexcoord{3.5*\n}{6.5*\n} circle (3pt) node[above, violet]{\scalebox{.8}{$c$}};
			
			\draw[very thick, color=blue,yshift=.2mm,xshift=-.15mm]\hexcoord{2.5*\n}{7.5*\n}(60:\n) node[blue]{$\bullet$}-- (0:\n)-- (300:\n);
			\draw[very thick, color=green,yshift=-.2mm,xshift=.15mm]\hexcoord{2.5*\n}{7.5*\n} (0:2*\n) node[green]{$\bullet$}-- (0:\n)-- (300:\n);
			\draw[thick, color=violet, fill=violet, shift=(0:\n)]\hexcoord{2.5*\n}{7.5*\n} circle (2.75pt) node[above, violet,xshift=.2mm]{\scalebox{.8}{$d$}};
			\draw[thick, color=blue, fill=green, shift=(300:\n)]\hexcoord{2.5*\n}{7.5*\n} circle (2.75pt);
			
		\end{tikzpicture}
		
	\end{center}
\caption{Polymers for case 3b.}
\end{subfigure}
\caption{All the possible types of polymers that contribute for cases 1-3 that pass through the vertices $a$-$d$. The contributing endpoints and edges to these polymers are colored. For case 1, odd length polymers must pass through a ``corner'' as illustrated by the polymers passing through vertex $b$. For case 3a, all polymers of length four contribute. For case 3b, the figure shows all length four polymers passing through vertices $a$ and $c$, all length three polymers passing through $b$, and all length 2 polymers passing through $d$.}
\label{fig:small_walks}
\end{figure}

\begin{proposition}\label{prop:mu}
	Let $d\geq 0$, and fix $0\leq \max\{K,1\} < N$. Then for any $\phi\in \mr{\caP}_{N,K}^{(d)}$,
	\begin{equation}\label{counting_bounds}
		n_k(\phi):=\left|\set{\phi'\in \mr{\caP}_{N,K}^{(d)}: \phi \disjoint \phi', \, \ell(\phi')=k}\right| \leq C_k(\ell(\phi)+1)\leq\ell(\phi)288^{k/5}
	\end{equation}
	for all $k\geq 1$ where
	\begin{equation}
		C_k = \sup_{0\leq K < N}\sup_{v\in\Gamma^{(0)}} \left| \set{ \phi \in  \mr{\caP}_{N,K}^{(d)} : \ell(\phi)=k, \deg_\phi (v) > 0} \right|.
	\end{equation}
\end{proposition}

\begin{proof}
	As in Lemma~\ref{lem:number-polymers}, it is sufficient to consider $d=0$. The first inequality in \eqref{counting_bounds} is trivial since any polymer has at most $\ell(\phi)+1$ vertices. The bound for $k\geq 5$ holds by Lemma~\ref{lem:number-polymers}, by observing that
	\[
	C_k(\ell(\phi)+1)\leq 3\left(\ell(\phi)+1\right)(k+1)2^{k-2} \leq \ell(\phi)\left( 2^{1-1/k}3^{1/k}(k+1)^{1/k} \right)^k \leq \ell(\phi) 288^{k/5}.
	\]
	where the final inequality holds since the expression inside the parentheses of the middle term is decreasing in $k$.
	
	The result is then a consequence of bounding $C_k$ for $k<5$. Only walks $\phi\in\mr{\caS}_{N,K}^{(0)}$ need to be considered as any loop has length $\ell(\phi)\geq 6$. The desired bound is obtained by considering separately the number of walks that intersect an arbitrary degree 3 vertex $v \in\Lambda^{(0)}_N \setminus \mr{\Lambda}^{(0)}_K$ for which:
	\begin{enumerate}
		\item both endpoints are in $\partial \Lambda^{(0)}_N$,
		\item both endpoints are in $\partial \Lambda^{(0)}_K$, and
		\item one endpoint is in $\partial \Lambda^{(0)}_K$ and the other is in $\partial \Lambda^{(0)}_N$.
	\end{enumerate}
Figure~\ref{fig:small_walks} illustrates the counting for each of the cases below. In each case, denote by $\delta_N$ the closed loop encompassing the border of $\Lambda_N^{(0)}$, i.e. the only closed loop in $\Lambda_N^{(0)}\setminus \Lambda_{N-1}^{(0)}$, and assume that $\phi \in \mr{\caS}_{N,K}^{(0)}$ with $\ell(\phi) = k \leq 4$. 
	
	Case (1): If both endpoints are in $\partial \Lambda^{(0)}_N$, then $\phi \subset \delta_N$ as any walk with an edge outside of $\delta_N$ has length at least 5. The distance between neighboring vertices in $\partial \Lambda^{(0)}_N$ is always two except for the six pairs of vertices at the ``corner'' edges of $\Lambda^{(0)}_N$, which have distance one. As a consequence, there are at most $k$ walks passing through a given degree 3 vertex $v$ since $N\geq 2$.
	
	Case (2): Recall that $\phi$ contains no edges from $\Lambda_K^{(0)}$ by definition of $\mr{\caP}_{N,K}^{(0)}$. If both endpoints of $\phi$ are in $\partial \Lambda^{(0)}_K$, it must be that $\ell(\phi)= 4$ and $\phi$ connects two neighboring vertices in $\partial \Lambda^{(0)}_K$. In this case, for any given vertex $v$ there are at most two such possible walks passing through $v$.
	
	Case (3): Suppose one endpoint of $\phi$ is in $\partial \Lambda^{(0)}_K$ and one is in $\partial \Lambda^{(0)}_N$. As the graph distance $D_0$ of $\Gamma^{(0)}$ satisfies
	\[D_0(\partial\Lambda_N^{(0)},\partial\Lambda_K^{(0)})=2(N-K) \]
	it must be that $N=K+1$ or $N=K+2$ (otherwise there are no such $\phi$).
	
	If $N=K+2$, then $\ell(\phi)=4$ and every vertex $v\in\Lambda_N^{(0)}\setminus \mr{\Lambda}_K^{(0)}$ satisfies
	\[
	D_0(v,\partial\Lambda_N^{(0)}) = 4 - D_0(v,\partial\Lambda_K^{(0)}), \qquad 0 \leq D_0(v,\partial\Lambda_K^{(0)}) \leq 4.
	\]
	As every walk through $v$ is the concatenation of a walk from $\partial\Lambda_N^{(0)}$ to $v$ and a walk from $v$ to $\partial\Lambda_K^{(0)}$, considering the five cases separately and counting the number of boundary vertices that are the appropriate distance to $v$ shows that there are at most four such walks passing through $v$.
	
	For the case $N=K+1$, call any edge with one vertex from $\partial \Lambda^{(0)}_K$ and one vertex from $\delta_N$ a \emph{bridge}. Then, any walk $\phi$ with length at most $4$ contains exactly one bridge, plus $\ell(\phi)-1$ edges in $\delta_N$. Therefore, if $k=1$ there are no such walks. If $k=2$ or $3$, there are at most two such walks passing through a fixed vertex, while if $k=4$ there are four: if $v$ is in $\partial \Lambda^{(0)}_N$, it can reach up to 4 bridges and there is exactly one $\phi$ that contains $v$ and each such bridge. If $v\not \in \partial \Lambda^{(0)}_N$ then $v$ belongs to a bridge, which belongs to at most two different walks of length $4$. Moreover there are at most two additional walks that contain $v$ and a neighboring bridge.
	
	Summing the contributions from all three cases produces an upper bound on $C_k$ of the form
	\[
	C_k \le \begin{cases}
		1 & \text{$k=1$} \\
		4 & \text{$k=2$}  \\
		5 & \text{$k=3$} \\
		10 & \text{$k=4$}
	\end{cases}
	\]
which is sufficient to verify that $n_k(\phi) \leq C_k(\ell(\phi)+1) < \ell(\phi)288^{k/5}.$
\end{proof}

\bibliographystyle{plain}

\vfill

\noindent Angelo Lucia\\
Departamento de Análisis y Matemática Aplicada,\\
Universidad Complutense de Madrid, 28040 Madrid, Spain, and\\
Instituto de Ciencias Matemáticas, 28049 Madrid, Spain\\
\verb+anglucia@ucm.es+\\

\noindent Alvin Moon\\
Centre for the Mathematics of Quantum Theory\\
Department of Mathematical Sciences, University of Copenhagen. \\
Universitetsparken 5, 2100 K{\o}benhavn, Denmark\\
\verb+alvin.s.moon@gmail.com +\\

\noindent Amanda Young\\
Department of Mathematics\\
University of Illinois Urbana Champaign\\
Urbana, IL, USA\\
\verb+ayoung86@illinois.edu+\\

\end{document}